\documentclass[12pt,reqno]{article}

\usepackage[all]{xy}           %For commutative diagrams
\usepackage{amssymb,amsmath}           %For double arrows
\usepackage{amsthm}
\usepackage{hyperref}
\usepackage{eucal}
\usepackage{nicefrac}
\usepackage{color}
\usepackage{mathrsfs}
\usepackage{graphicx}
\parskip=\medskipamount
\newtheorem{definition}{\noindent \noindent {\bf
Definition}}[section]
\newtheorem{lem}{{\bf Lemma}}[section]

\newtheorem{prop}{{\bf Proposition}}[section]

\newtheorem{remark}{{\bf Remark}}[section]

\def\r{\ensuremath{\mathbb{R}}}

\def\bcr{\begin{color}{red}}
\def\bcb{\begin{color}{blue}}
\def\bcg{\begin{color}{green}}
\def\enc{\end{color}}

 \def\ba{\begin{array}}
  \def\ea{\end{array}}

\def\rb {\mathcal{R}}

\def\R{\mathbb{R}}
\def\Li{\mathscr{L}}

\def\d{{\rm d}}

\def\derpar#1#2{\ds\frac{\partial{#1}}{\partial{#2}}}

\def\bea{\begin{eqnarray}}
\def\eea{\end{eqnarray}}
\def\beq{\begin{equation}}
\def\eeq{\end{equation}}
\def\beann{\begin{eqnarray*}}
\def\eeann{\end{eqnarray*}}
\newcommand{\ds}{\displaystyle}
\renewcommand{\neq}{=\hspace{-3.5mm}/\hspace{2mm}}

\def\bcr{\begin{color}{red}}
\def\bcb{\begin{color}{blue}}
\def\enc{\end{color}}

\def\fpd#1#2{{\frac{\partial #1}{\partial #2}}}

\def\vectorfields#1{{\mathfrak X}(#1)}
\def\cinfty#1{C^{\scriptscriptstyle\infty}(#1)}
\def\g{\mathfrak{g}}

 \definecolor{ochre}{rgb}{0.8, 0.47, 0.13}

\title{\sf Symmetry reduction  and reconstruction in  contact geometry and Lagrange-Poincar\'e-Herglotz equations}

\author{\sffamily 
Alexandre Anahory Simoes$^1$,\thanks{alexandre.anahory@ie.edu\quad ORCID: 0000-0003-4644-876X}
Leonardo Colombo$^2$,\thanks{leonardo.colombo@csic.es\quad ORCID:0000-0001-6493-6113 }
\\[0.1ex]
\sffamily 
Manuel de Le\'on$^{3,4}$, \thanks{mdeleon@icmat.es\qquad\qquad ORCID: 0000-0002-8028-2348}\, 
Modesto Salgado$^{5,6}$\thanks{modesto.salgado@usc.es\quad ORCID: 0000-0003-3982-1845}\,  
and 
Silvia Souto$^{5}$\thanks{silviasouto.perez@usc.es\quad ORCID: 0000-0003-0755-1211}
\\[1ex]
\\[0.1ex]
\normalsize\itshape\sffamily 
$^1$ IE School of Science and Technology, Madrid, Spain
\\[0.1ex]
\normalsize\itshape\sffamily 
$^2$Centre for Automation and Robotics (CSIC-UPM), Arganda del Rey, Spain
\\[0.1ex]
\normalsize\itshape\sffamily 
$^3$Instituto de Ciencias Matemáticas (CSIC), Madrid, Spain
\\[0.1ex]
\normalsize\itshape\sffamily 
$^4$Real Academia de Ciencias, Madrid, Spain
\\[0.1ex]
\normalsize\itshape\sffamily 
$^5$Departamento de Matemáticas, Facultade de Matemáticas,
\\[0.1ex]
\normalsize\itshape\sffamily 
Universidade de Santiago de Compostela, Spain
\\[0.1ex]
\normalsize\itshape\sffamily
$^6$Centro de Investigación y Tecnología Matemática de Galicia (CITMAga), Spain
}

\begin{document}

%%%%% Portada %%%%%%%%%%%%%%%%%%%%%%%%%%%%%%%%%%%%%%%%%
\maketitle

%\tableofcontents

%\pagestyle{plain}
%\pagestyle{myheadings}
\date{}

\begin{abstract}
In this paper, we investigate the reduction process of a contact Lagrangian system whose Lagrangian is invariant under a group of symmetries. We give explicit coordinate expressions of the resulting reduced differential equations, the so-called Lagrange-Poincar\'e-Herglotz equations. Our framework relied on the associated Herglotz vector field and its projected vector field, and the use of well-chosen quasi-velocities. Some examples are also discussed.

%{\bf
%\bcb Trabajo de referencia \enc }

%\bcr T.\ Mestdag and M.\ Crampin, Invariant Lagrangians, mechanical connections and the Lagrange-Poincar\'{e} equations, J.\ Phys.\ A: Math.\ Theor. \ 41 (2008) 344015 (20pp).
%\enc

  \end{abstract}



%%%%% \'{I}ndice %%%%%%%%%%%%%%%%%%%%%%%%%%%%%%%%%%%%%%%%%%

%\setcounter{tocdepth}{1}% inserta todos los ep\'{\i}grafes hasta el nivel \paragraph en la tabla de contenidos

% \tableofcontents % Insertar tabla de contenidos

%\renewcommand{\baselinestretch}{1.5}
%\renewcommand{\arraystretch}{0.66}

%%%%% Cuerpo del documento %%%%%%%%%%%%%%%%%%%%%%%%%%%%

\section{Introduction}\label{section 1}

 Lagrangian and Hamiltonian contact systems have been the subject of intense activity in recent years. The fundamental difference with traditional mechanics is that, instead of possessing conservative properties, contact systems exhibit dissipation. The geometrical reasons for this different behaviour is that one case uses symplectic geometry (a closed, non-degenerate 2-form) while the latter case involves contact forms (a non-degenerate 1-form that is not closed). From the point of view of variations, symplectic mechanics is based on Hamilton's principle, while contact mechanics is based on Herglotz's principle, a generalisation that requires solving an implicit differential equation to define the action \cite{herglotz,bravetti,MdLMLV,canarios}. These Lagrangian functions are commonly called in Physics as action-dependent Lagrangians.

One of the most relevant applications of differential geometry to the study of mechanical systems are the procedures for the reduction of the dynamics when the system exhibits symmetries. The mechanisms are: the symplectic reduction theorem of Marsden and Weinstein via the momentum application (an extension of the renowned Noether theorem \cite{mw,AM}) and its generalisation to the contact case \cite{Willett,Albert,MdLMLV}; the coisotropic reduction which in the symplectic case was developed by Weinstein \cite{Weinstein,AM,LR}, and in the contact case by de Le\'on and Lainz \cite{MdLMLV} (see also the recent survey \cite{ruben}); and the direct reduction by the Lie group of symmetries, which in the symplectic case has been developed by Mestdag and Crampin \cite{MC,RRA}. An alternative variational approach is the so-called Lagrangian reduction by stages developed by Cendra, Marsden and Ratiu \cite{Cendra}.%. This work presents a complete exposition of the reduction of Lagrangian mechanical systems, a process of repeated reduction \cite{Cendra}.

 In a recent work \cite{A1} we have obtained the Euler–Poincaré–Herglotz equations and Lagrange–Poincaré–Herglotz equations, respectively (correspondingly Lie–Poisson–Jacobi equations and Hamilton–Poincaré–Herglotz equations in the Hamiltonian side) when the contact Lagrangian system is considered on a Lie algebroid, and, in particular a Atiyah algebroid if we are in presence of symmetries. Similar results are obtained in \cite{A2}, but in this case we are taking the methodology of prolongations of Lie algebroids, which permits to treat directly with contact structures and not general Jacobi structures as above. In \cite{EPH}, the authors discussed reduction procedures for contact mechanical systems on Lie groups based on variational principles.

In this paper, we develop a procedure for the reduction and reconstruction of the dynamics of Lagrangian contact systems using the quasi-coordinate method, extending the results of \cite{MC} for conservative systems to systems that include dissipation. One of the necessary geometrical constructions is the prolongation of a connection in a principal fiber bundle using the Hessian of the Lagrangian (instead of a Riemannian metric as in the case of Marsden \cite{book}).
The corresponding reduced dynamics are governed by the so-called Lagrange-Poincar\'e-Herglotz equations.

The paper is structured as follows. Section 2 is devoted to introduce some concepts on the geometry of the fiber bundle $TQ \times \mathbb{R}$, including Second Order Differential Equations (SODE for short), local frames and quasi-velocities. The main notions on contact Lagrangian systems are developed in Section 3, where the dynamics is provided by the so-called Herglotz SODE. In Section 4 we consider configuration manifolds that are principal bundles over a quotient $Q \to Q/G$, where $G$ is a Lie group acting on $Q$; we also discuss connections and $G$-invariant vector fields, in order to get $G$-invariant local frames. Sections 5, 6 and 7 contain the main results of this paper. Indeed, in Section 5 we give a reduction procedure of the original dynamics to the space $TQ/G \times \mathbb{R}$, and so we obtain the so-called Lagrange-Poincar\'e-Herglotz equations. Section 6 is an auxiliary section for Section 7; in fact, in order to get a reconstruction procedure we need to construct a connection in the principal bundle
$TQ \times \mathbb{R} \to TQ/G \times \mathbb{R}$ using the nondegeneracy of the Lagrangian. We also include one example to illustrate our results as well as a final section of conclusions and further work.

\section{Preliminaires}
\subsection{Connections and curvature}

In what follows we will often use non-linear connections on different many fiber bundles. In order to unify notations, let us recall briefly some definitions. Let $p: E \to B$ be a fiber bundle, that is, $p$ is a surjective submersion. For $e\in E$, the vertical space $V_eE$ at $e$ is the kernel of $T_ep: T_eE \to T_{p(e)}B$. So, we obtain the so-called \textbf{vertical distribution} $VE=\{V_eE  | e\in E\}$. Next, we construct a short exact sequence of vector bundles over $E$ as follow,
\begin{equation} \label{seq}
0 \to VE \to TE \to E\times_B TB \to 0,
\end{equation}
where the middle arrow $j: TE \to E\times_B TB$ is given by $v_e \mapsto (e, Tp(v_e))$, and $E\times_B TB$ denotes the fiber product of $E$ and $TB$ over $B$.

A \textbf{connection} on $p$ is either given by a {right splitting $\gamma: E\times_B TB \to TE$ }(i.e.\ a linear  map satisfying $j \circ \gamma = id_{E\times_B TB }$), or by the corresponding left splitting $\omega = id_{TE} -\gamma    \circ j: TE \to VE\subset TE$.

The above short exact sequence naturally extends to the level of sections of the corresponding bundles over $E$,
\[
0 \to Sec(VE) \to \vectorfields{E} \to Sec(E\times_B TB) \to 0.
\]
The notation $\vectorfields{E}$ stands for the Lie algebra of vector fields on $E$.
A splitting of (\ref{seq}) induces a splitting of the second sequence, and conversely. When we interpret $\omega: \vectorfields{E} \to \vectorfields{E}$ as a (1,1)-tensor field on $E$, we will call it the \textbf{connection form}, or the \textbf{vertical projection}. The map $h:=id-\omega$ is the \textbf{horizontal projection} of the connection. Since vector fields $T$ on $B$ can be thought as basic sections in $Sec(E\times_B TB)$, we may define the \textbf{horizontal lift} of the vector field $T$ as the vector field $T^h$ of $E$, given by $T^h(e) = \gamma(e,T(\pi(e)))$, for each $e\in E$.

The \textbf{curvature} of the connection is the (1,2)-tensor field on $E$, given by $(X,Y) \mapsto -\omega([hX,hY])$, for two vector fields $X,Y\in\vectorfields{E}$. In what follows, however, we will also often use the word `curvature' for the restriction of that map to two horizontal lifts and use the notation
\[
K(T,S) = - \omega([T^h,S^h])  \in\vectorfields{E}
\]
when  $T,S\in\vectorfields{B}$.

Assume now that $p: E \to B$ is a principal bundle. This means that there exists a Lie group $G$ acting on $E$ on the left such that $B$ is naturally identified with the quotient manifold $E/G$. In addition, the fiber bundle is locally trivial, in the sense that for any point in $B$ there exists a neighborhoopd $U$ such that
$\pi^{-1} (U)$ is diffeomorphic to the product $U \times G$ and this diffeomorphism preserves the action of $G$ (the action of $G$ on $U \times G$
is the trivial one). We know that for any element $\xi$ of the Lie algebra $\mathfrak{g}$ of $G$ there exists a fundamental vector field $\xi_E$ on $E$ provided by the exponential of the group.

Consider now the left splitting $\omega_e: T_eE \to V_eE$ as in (\ref{seq}), for any $e \in E$. Since
$$
\omega_e(\xi_E)(e) = \xi_E(e)
$$
we can define a mapping 
$$
\overline{\omega} : T_eE \to \mathfrak{g}
$$
by
$$
\overline{\omega}_e (\xi_E(e)) = \xi
$$
This mapping is called the connection form, it is a 1-form taking values in $\mathfrak{g}$ and if it has the property that
$$
\Phi_g^* \overline{\omega} = Ad(g) \overline{\omega}
$$
for any $g \in G$, where $\Phi_{g}$ is the left action of the Lie group element $g$ and $Ad$ is the adjoint representation of the Lie group $G$, then it is called a principal connection.

\subsection{Geometric structures on \(TQ\times\r\)}\label{sec2}

%\subsection{The tangent bundle of \texorpdfstring{$k^1$}--velocities}
%  

Let $Q$ be a manifold with dimension $n$ with local coordinates $(q^\alpha)$.  Let $\tau_Q: TQ \to Q$ be the tangent bundle of a smooth manifold $Q$. If $(q^\alpha)$ are local coordinates on $U \subseteq Q$, the natural coordinates $(q^\alpha, u^\alpha)$ on $T U=\tau_Q^{-1}(U)$ are
$$
 q^\alpha(v_q)=q^\alpha(q),\qquad
  u^\alpha(v_q)=v_q(q^\alpha) \,.
  $$
 
 \paragraph{Lifts of functions.}
 
If $f$ identificación a differentiable function on $Q$, the \textit{vertical lift} $f^V$ and the \textit{complete lift} $f^C$ of $f$, are the functions on 
$ TQ $ given by
$$f^V(v_q)=(\tau_Q)^*f(v_q)=f(q)
\,, \qquad f^{C}(v_q)=v_{q}(f)\equiv u^\alpha(v_q)\ds\derpar{f}{q^\alpha}\Big\vert_q\,  .$$

Since $(q^\alpha)^V=q^\alpha$ and $(q^\alpha)^{C}=u^\alpha$, we deduce that vector fields on $ TQ $ are characterized by its action on vertical and complete lifts of functions.

\paragraph{Lifts of vector fields.} Given a vector field $X\in  \mathfrak{X}(Q)$, the \textit{vertical lift} $X^{V}$, and the \textit{complete lift} $X^C$ are the vector fields on $T Q$ given by 
\begin{equation}\label{VClifts}
\begin{array}{lcl}  
X^{V}(f^V)=0, & &X^{V}(f^C)=  (X(f))^V ,
\\ \noalign{\medskip}
X^C(f^V)=(X(f))^V, & & X^C(f^C)=(X(f))^{C}, 
 \end{array}
\end{equation}
for any function $f$ on $Q$.

Taking adapted coordinates $(q^\alpha, u^\alpha)$ on $T Q$, if $X=X^\alpha \ds\derpar {}{q^\alpha}$, then
\begin{equation}\label{lifts-vectors1} 
X^{V} =X^\alpha \derpar{}{u^\alpha} \,,
\qquad 
X^C=X^\alpha \derpar{}{q^\alpha}+ u^\beta\derpar{X^\alpha}{q^\beta}\derpar{}{u^\alpha} \,,
 \end{equation} 
and we have
\begin{equation}\label{lifts-vectors2} 
\left(\derpar{}{q^\alpha}\right)^{V} =   \derpar{}{u^\alpha} \,,
\qquad 
\left(\derpar{}{q^\alpha}\right)^C=  \derpar{}{q^\alpha}\,  .\end{equation}
Moreover, we have
\begin{equation}\label{lifts-vectors13}
 [X^C,Y^C]=[X,Y]^C,\quad [X^C,Y^V]=[X,Y]^V\, .
\end{equation}
and
\begin{equation}\label{lifts-prod-funct-vector}
 (f\, X)^V=f^V\,X^V, \quad ,\quad (f\, X)^C=f^V\,X^C+f^C\, X^V\, .
\end{equation}
 
The \textbf{Liouville vector field} $\Delta\in \mathfrak{X}(
T Q)$ is the 
vector field defined by $\Delta(f^V)=0$ and $\Delta (f^C)=f^C  .$
In adapted coordinates, it reads
\begin{equation}\label{delta1} \Delta = u^\alpha \derpar{}{u^\alpha} .\end{equation}

\begin{remark}\
{\rm 
\begin{enumerate} %[{\rm (1)}]
\item If $X$ is a vector field on $Q$ with local one-parameter group $\phi_t: Q \to Q$, then $X^C$ is the infinitesimal generator of the flow $T\phi_t : TQ \to TQ$.

\item $\Delta$ is the infinitesimal generator of the flow 
$\phi : \R \times TQ  \to TQ  \,,\, \phi(t,v_q) = e^t v_q,   $.
\end{enumerate}
}\end{remark}

As a consequence of \eqref{lifts-vectors1} and  \eqref{lifts-vectors2} we deduce that  the   tensor fields of type $(1,1)$ on $ TQ $ are characterized by the action on these lifts of vector fields.

The \textbf{canonical tangent structure} of $TQ$ is the   tensor field $S$ 
of type $(1,1)$  defined by
$S(X^C)=X^V$ and $S(X^V)= 0 \, .$
In local adapted coordinates it is written as
\begin{equation}\label{localJ}
S=\displaystyle\frac {\displaystyle\partial}{\displaystyle\partial u^\alpha} \otimes \d q^\alpha \,
.
\end{equation}

\begin{remark} \label{ext}{\rm Since $TQ \times\R \to TQ $ is a trivial vector bundle, 
the lifts of functions, vector fields, and  the canonical structures on $TQ $ (the 
canonical tangent structure and the
Liouville vector field described above)
can be extended to $TQ \times\R $ in a natural way,
and are denoted with the same notations ($f^V$ ,$f^C$, $X^V$, $X^C$, $S$ and $\Delta$).}
\end{remark}

 %\tom{name of the file:} \verb!NEWVERSION2.tex!.  \tom{Version of \today. }
 
%\begin{equation}\label{tau}
% \mathfrak{X}ymatrix{
%	& TQ \times \r \ar[dl]_{\tau_1} \ar[d]^{\tau_0} \ar[dr]^s \\
%	TQ & Q \times \r & \r
%}
%\end{equation}	
%Notice that $\tau_1$ and $\tau_0$ are the projection maps of two different vector bundle structures. We will usually have the second one in mind. In fact, with this structure, $TQ \times \r$ is the pull-back of the tangent bundle $TQ$ with respect to the projection $Q \times \r \to Q$. We will denote by $\left(q^\alpha,u^\alpha,s\right)$ the natural coordinates in $TQ \times \r$.
%
%If $\left(v_q,s\right) \in TQ \times \r$, then
%$$\tau_1\left(v_q,s\right)=v_q \,,\;\;\, \tau_0\left(v_q,s\right)=\left(q,s\right) \,,\;\;\, s\left(v_q,s\right)=s .$$
%
%We can write $T\left(TQ \times \r\right) = \left(T(TQ) \times \r\right) \oplus \left(TQ \times T\r\right)$ and hence every operation acting on tangent vectors of $TQ$ can act on tangent vectors of $TQ \times \r$. So in the next paragraphs we briefly recall some canonical objects and structures that can be defined on $TQ \times \r$. 

\subsection{Second-order differential equations}\label{kvf}
%%%%%%%%%%%%%%%%%%%%%%%%%

\begin{definition}\label{dfn:k-contact-holonomic-section}
    Consider a map $c_s\colon I \subset\r \to Q\times\r$ with
    $ c_s(t)= (c(t),s(t))\,, $
    where $c\colon I \subset\r  \to Q$. The \textbf{first prolongation} of $c_s$ to $T Q \times\r$ is the curve $c_s^{(1)}\colon I \subset \r\to T Q \times\r$ given by
    $$ c_s^{(1)}(t)= (\dot c(t),s(t))\,. $$

In local coordinates, if   $ c_s(t)= (c(t),s(t))=(q^\alpha(t),s(t))\,, $
 then we have
\begin{equation}\label{phis1}
     c_s^{(1)}(t)= (\dot c(t),s(t))  = \left(q^\alpha(t), \dot q^\alpha(t),s(t)\right).
\end{equation}	
	
%	A vector field $\Gamma \in \mathfrak{X}(TQ \times \r)$ is said to satisfy the second-order condition (for
%	short, is a SODE) when all of its integral curves are holonomic.
\end{definition}

\begin{definition}\label{xijso}
    A vector field $\mathbf{\Gamma}$ on $   T  Q \times\R $ is a \textbf{second-order differential equation} (or a {\sc \textbf{SODE}}) if $S(\Gamma)=\Delta$.
\end{definition}

The local expression of a SODE is
\begin{equation}
	\label{localsode1} \Gamma  = u^\alpha  \frac{\partial}
	{\partial   q^\alpha}+ f^\alpha  \frac{\partial} {\partial
		u^\alpha }+ g  \frac{\partial} {\partial
		s }, 
 \end{equation}
for some local functions $f^\alpha$, $g$  $\in \cinfty{TQ \times \r}$.

If $\sigma\colon I\subset \R  \to T  Q\times\R $, locally given by
$\sigma(t)=(q^\alpha(t),u^\alpha(t),s(t))$, is an integral curve of a {\sc sode} $\mathbf{\Gamma} $, from Definition \ref{dfn:k-contact-holonomic-section} and equation \eqref{localsode1}, it follows that
\begin{equation}\label{solsopde}
    \ds\frac {d q^\alpha}{dt}\Big\vert_t = u^\alpha(t)\,,\qquad \ds\frac  {du^\alpha}{dt}\Big\vert_t = f^\alpha(\sigma(t))\,,\qquad \ds\frac  {ds}{dt}\Big\vert_t = g^\alpha(\sigma(t))\,.
\end{equation}
Then, we have
$$
    \sigma(t) = (q^\alpha(t),\ds\frac {d q^\alpha}{dt}\Big\vert_t,s(t))        
    \quad \ds\frac  {ds}{dt}\Big\vert_t = g^\alpha(\sigma(t))\,.
$$
and hence $\sigma(t)= c_s^{(1)}(t)= (\dot c(t),s(t))$ where $c(t)=\tau \circ \sigma(t)=
(q^\alpha(t))$ is the \textbf{first prolongation} of $c_s(t)=(  c(t),s(t))$ to $T Q \times\r$.

So, in coordinates a  {\sc sode} defines a system of differential equations of the form

\begin{equation}
	\label{nn1} \frac{d^2 q^\alpha} {dt^2}=f^\alpha(q,\dot{q},s) ,\quad  \frac{ds} {dt}=g(q,\dot{q},s).
\end{equation}

\subsection{Local frames and quasi-velocities}

Consider $\{Z_{\alpha}\}$ a \textbf{local frame} on $Q$, that is, a local basis of vector fields on $Q$. Then each vector field $Z \in \mathfrak{X}(Q)$ can be written as $$Z = Z^{\alpha} Z_{\alpha}$$ for some local functions $Z^{\alpha} \in C^\infty(Q)$.

Likewise, each tangent vector $w_q \in T_q Q$ can be descomposed as $w_q = v^\alpha(w_q) Z_{\alpha} (q)$, for some real numbers $v^\alpha(w_q)$. 

This new fibre coordinates $v^\alpha$, which are the components of $w_q$ with respect to the basis $Z_{\alpha}(q)$, are called \textbf{quasi-velocities}.

If we consider the standard frame $\{\partial/\partial q^\beta\}$ on $Q$, then $$Z_{\alpha} (q) = Z^{\beta}_{\alpha} (q)\ds\frac{\partial}{\partial q^{\beta}}\Big\vert_q.$$
Therefore,
$$w_q = v^\alpha(w_q) Z_{\alpha} (q) = v^\alpha(w_q) Z^{\beta}_{\alpha} (q)\ds\frac{\partial}{\partial q^{\beta}}\Big\vert_q ,$$
and we have $$u^\beta( w_q)= v^\alpha(w_q) Z^\beta_\alpha(q)\,.$$
So, we can use $\left(q^\alpha,v^\alpha,s\right)$ as (non-natural) coordinates in $TQ \times \r$,
$$
v^\alpha(w_q)= (Z_\beta^\alpha(q))^{-1} u^\beta( w_q)$$

From the local expressions (\ref{lifts-vectors1}), we can easily conclude that:
\begin{prop}\label{ReflocTQR}
If $\{Z_{\alpha}\}$ is any local frame on $Q$, then $\{Z_{\alpha}^C,Z_{\alpha}^V, \frac{\partial}{\partial s}\}$ is a local frame on $TQ \times \r$.
\end{prop}

\begin{lem}
We may express (locally) a second-order differential equation (SODE) $\Gamma \in \mathfrak{X}\left(TQ \times \r\right)$ in the form
\begin{equation}\label{sodelf}
\Gamma = v^\alpha Z^C_\alpha + \Gamma^\alpha Z^V_\alpha + g \ds\frac{\partial}{\partial s} ,
\end{equation}
where $v^\alpha$ are the quasi-velocities, and $\Gamma^\alpha$, $g$ are some local functions on $ TQ \times \r $.
\end{lem}

The proof is a consequence of the above Proposition \ref{ReflocTQR} and the following properties: $T \tau_Q \circ Z_\alpha^C = Z_\alpha \circ \tau_Q$ and $T \tau_Q \circ Z_\alpha^V = 0$, where 
$  \tau_Q: TQ \times \r \to Q $ is the canonical projection.

\section{Contact Lagrangian systems }

In this section, we recall the Lagrangian contact formalism and we introduce the notion of $G$-invariant vector fields associated with a connection. For a regular Lagrangian, the solution of the Herglotz equations is given by the integral curves of some vector field, the so-called \textbf{Herglotz SODE}.
 
A \textit{contact manifold} is a pair $(M,\eta)$ where $M$ has odd dimension $2n+1$ and $\eta$ is a contact $1$-form, that is, $\eta \wedge d\eta^n \neq 0$. On these manifolds, there exist a unique vector field (called \textit{Reeb vector field}) $\rb $ such that $i_\rb  \eta = 1$ and $i_\rb  d\eta = 0$.

 Consider now a Lagrangian function $L: TQ \times \r \to \r$, $L=L(q^\alpha,u^\alpha,s)$.
We will assume that $L$ is regular. That is, its Hessian matrix with respect to the velocities $(W_{\alpha\beta})$,  where
$$W_{\alpha\beta}=\frac{\partial^2 L}{\partial u^\alpha \partial u^\beta}, $$
 is regular.
 
Another way to read the regularity of the Lagrangian $L$ in terms of a non-standard frame is the  following:  {\it  a Lagrangian $L$ is regular if, and only if, the matrix
$$\left(Z_\alpha^V\left(Z_\beta^V\left(L\right)\right)\right)$$
has maximal rank.}

 We are able to define the Lagrangian $1$-form  on $TQ \times \r$
\begin{equation}\label{1forcont}
	\eta_{L}=ds - \theta_{L}= ds-dL\circ S=ds-\ds\frac{\partial L}{\partial u^\alpha} \, dq^\alpha 
\end{equation}
which is a contact form on $TQ \times \r $ if and only if $L$ is regular. Indeed,
\begin{equation}\label{detal}
	d\eta_{L}= - \ds\frac{\partial^2 L}{\partial s \partial u^\alpha} \, ds \wedge dq^\alpha - \ds\frac{\partial^2 L}{\partial q^\beta \partial u^\alpha} \, dq^\beta \wedge dq^\alpha - \ds\frac{\partial^2 L}{\partial u^\beta \partial u^\alpha} \, du^\beta\wedge dq^\alpha .
\end{equation}
and then
$$ \eta_L\wedge   d\eta_L^n = det(W_{\alpha\beta}) d^nq \wedge   d^nu \wedge ds\neq 0.$$

The triple $\left(TQ \times \r, \eta_{L},E_{L}\right)$ is said to be a \textit{contact Lagrangian system}.
 The corresponding Reeb vector field, determined by the relations
$i_{\rb _{L}} \eta_{L} = 1 \,\,,  i_{\rb _{L}} d\eta_{L} = 0 ,$ 
 is locally given by 
\begin{equation}\label{RbL}
	\rb _{L}= \ds\frac{\partial}{\partial s} - W^{\alpha \beta} \ds\frac{\partial^2 L}{\partial s \partial u^\beta} \ds\frac{\partial}{\partial u^\alpha} ,
\end{equation}

The Lagrangian energy of the system is defined as
\begin{equation}\label{enerlagcont}
		E_{L}=\Delta(L)-L=u^\alpha\, \ds\frac{\partial L}{\partial u^\alpha} - L .
	\end{equation}
	
	From a straightforward computation, we obtain
	$$\rb _{L} \left(E_{L}\right) = - \ds\frac{\partial L}{\partial s} . $$

	We denote by $\flat_L$ the vector bundle isomorphism   for the contact
form $\eta_L$ on $TQ \times \r$. That is,
	\begin{equation}\label{isobL}
	\begin{array}{lcll}
		\flat_{L}:& \mathfrak{X}\left(TQ \times \r\right) &\longrightarrow & \Omega^{1}\left(TQ \times \r\right) \\ \noalign{\medskip}
		& Z & \longmapsto &  \flat_{L}(Z)=i_{Z}d\eta_{L}+\eta_{L}(Z)\eta_{L}.
	\end{array}
\end{equation}

The dynamics of the system is given by the Herglotz vector field $\Gamma$, which  is the unique vector
field satisfying
\begin{equation}\label{Gamma0}
	\flat_{L}\left(\Gamma\right)=dE_{L}+ \left(
	 \ds\frac{\partial L}{\partial s}  - E_{L}\right) \eta_{L} .
\end{equation}
A direct computation from equation (\ref{Gamma0})
	shows that $\Gamma$ is a {\sc sode} locally given by
	\begin{equation}\label{elocSODE}
\Gamma= u^\alpha \ds\frac{\partial}{\partial q^\alpha} + \Gamma^\alpha \ds\frac{\partial}{\partial u^\alpha} + L \ds\frac{\partial}{\partial s} ,
\end{equation}
where the components $\Gamma^\alpha$ satisfy the equation
\begin{equation}\label{Bi}
	\Gamma^\beta \ds\frac{\partial^2 L}{\partial u^\beta \, \partial u^\alpha} + u^\beta \ds\frac{\partial^2 L}{\partial q^\beta \, \partial u^\alpha} + L \, \ds\frac{\partial^2 L}{\partial s \, \partial u^\alpha} - \ds\frac{\partial L}{\partial q^\alpha}= \ds\frac{\partial L}{\partial s} \ds\frac{\partial L}{\partial u^\alpha} .
\end{equation}

	 Let $\left(q^\alpha(t),u^\alpha(t),s(t)\right)$ be an holonomic integral curve of $\Gamma$, that is $u^\alpha(t)= \ds\frac{dq^\alpha}{dt}\Big\vert_{t}$, then by substituting its values in equation (\ref{Bi})  we obtain  	 \begin{equation}\label{ELG}
\ds\frac{d^2 q^\beta}{dt^2} \ds\frac{\partial^2 L}{\partial u^\beta \, \partial u^\alpha} + \ds\frac{dq^\beta}{dt} \ds\frac{\partial^2 L}{\partial q^\beta \, \partial u^\alpha} + \ds\frac{ds}{dt} \, \ds\frac{\partial^2 L}{\partial s \, \partial u^\alpha} - \ds\frac{\partial L}{\partial q^\alpha}= \ds\frac{\partial L}{\partial s} \ds\frac{\partial L}{\partial u^\alpha} , \,\, \qquad  \ds\frac{ds}{dt}  =  L .
\end{equation}
which are the \textit{Herglotz equations}    considered by Herglotz in $1930$ \cite{herglotz}, usually written as
 \begin{equation}\label{EH}
	\ds\frac{d}{dt} \left(\ds\frac{\partial L}{\partial u^\alpha}\right) - \ds\frac{\partial L}{\partial q^\alpha}  =  \ds\frac{\partial L}{\partial s} \ds\frac{\partial L}{\partial u^\alpha} \,\,, \qquad  \ds\frac{ds}{dt}  =  L .
\end{equation}

The {\sc sode} $\Gamma \in \mathfrak{X}(TQ \times \r)$, solution to (\ref{Gamma0}), is called the \textit{Herglotz vector field} associated with the Lagrangian function $L$.

\bigskip

Next, we will show an alternative form to write Herglotz equations.
\begin{prop}\label{frameprop}
Let $\Gamma \in \mathfrak{X}\left(TQ \times \r\right)$ be a SODE and $L$ a regular Lagrangian on $TQ \times \r$. The following conditions are equivalent:
\begin{enumerate}
\item $\Gamma$ is Lagrangian vector field ,
\item\label{SODEr} $\Li_{\Gamma} \eta_{L} =\derpar{L}{s}  \eta_{L}\, , \quad \eta_{L}(\Gamma) = - E_L ,$
\item For each vector field $Z \in \mathfrak{X}(Q)$,
$$\Gamma\left(Z^V\left(L\right)\right) - Z^C\left(L\right) =  \derpar{L}{s}  \left(Z^V\left(L\right)\right) ;$$
or, equivalently, for each local frame $\{Z_{\alpha}\}$ of vector fields on $Q$,
\begin{equation}\label{ecsHlocframe}
\Gamma\left(Z_{\alpha}^V\left(L\right)\right) - Z_{\alpha}^C\left(L\right) =  \derpar{L}{s}  \left(Z_{\alpha}^V\left(L\right)\right) , \quad \alpha = 1,\ldots,n .
\end{equation}
\end{enumerate}
In particular, if we take the standard frame $\{\partial/ \partial q^\alpha\}$ on $Q$, equations (\ref{ecsHlocframe}) become
$$\Gamma\left(\ds\frac{\partial L}{\partial u^\alpha}\right) - \ds\frac{\partial L}{\partial q^\alpha} = \ds\frac{\partial L}{\partial s} \ds\frac{\partial L}{\partial u^\alpha} .$$
\end{prop}

\begin{proof}

$1. \Rightarrow 2.$ Assume $\Gamma$ is the unique solution of (\ref{Gamma0}) whose local expression is given in (\ref{elocSODE}). Then, by the isomorphism (\ref{isobL}) we have
$$i_\Gamma d\eta_{L} + \eta_{L}(\Gamma) \eta_{L} = dE_L - \left(\rb _L(E_L) + E_L\right)$$
i.e.
$$i_\Gamma d\eta_{L} -dE_L =  -\derpar{L}{s}  \eta_{L} ,$$
where we used that $\eta_{L}(\Gamma) = - E_L$.
Therefore,
$$\Li_{\Gamma} \eta_{L} = i_\Gamma d\eta_{L} + di_\Gamma \eta_{L} = i_\Gamma d\eta_{L} - dE_L =  -\derpar{L}{s}  \eta_{L} .$$

$2. \Rightarrow 3.$ If we apply the relation in \ref{SODEr} to a complete lift $Z^C$ we obtain
$$
\begin{array}{rcl}
0 & = & \left(\Li_{\Gamma} \eta_{L} - \ds\frac{\partial L}{\partial s} \eta_{L}\right) \left(Z^C\right) \\ \noalign{\medskip}
& = & \Gamma \left(\eta_{L}\left(Z^C\right)\right) - \eta_{L} \left([\Gamma, Z^C]\right) - \ds\frac{\partial L}{\partial s} \eta_{L} \left(Z^C\right) \\ \noalign{\medskip}
& = & - \Gamma \left(Z^V\left(L\right)\right) + Z^C\left(L\right) + \ds\frac{\partial L}{\partial s} \left(Z^V\left(L\right)\right) .
\end{array}
$$

When applied to a vertical lift $Z^V$ or to $\partial/\partial s$ we obtain the identity $``0=0"$. So the result follows.

The implications $2. \Rightarrow 1.$ and $3. \Rightarrow 2.$ are easily proved.
\end{proof}

%\section{ G-invariance}\label{sec4}

\section{ G-invariant vector fields}\label{sec4}
  
Throughout the paper we will assume that the configuration space $Q$ of the Lagrangian system is equipped with a free and proper left Lie group action $\Phi: G \times Q \to Q$ of a connected Lie group $G$, which eventually will be the symmetry group of the Lagrangian system under consideration.

The projection $$\pi_Q : Q \to Q/G $$ on the set of equivalence classes gives $Q$ the structure of a principal fibre bundle with structure group $G$.
 
Let $\Phi_g : Q \to Q$ and $\Phi_q : G \to Q$ be the maps
$$\Phi_g(q)=\Phi(g,q)=g\,q, \quad \Phi_q(g)=\Phi(g,q)=g\,q\, .$$

\begin{definition}
A vector field $W \in \mathfrak{X}(Q)$ is said to be \textit{G-invariant} if
$$\left(\Phi_g\right)_* (q) \left(W(q)\right) = W \left(\Phi_g(q)\right) ,$$
for all $g \in G$, $q \in Q$.
\end{definition}	
In that case, the relation
\begin{equation}\label{defred}
\breve{W} \circ \pi_Q=T\pi_Q \circ W
\end{equation}
uniquely defines a \textit{reduced vector field} $\breve{W} \in \mathfrak{X}(Q/G)$.

That is,
$$
\breve{W} \left(\pi_Q(q)\right) = \left(\pi_Q\right)_* (q) \left(W(q)\right) , \quad q \in Q .
$$

Likewise, if $F : Q \to \r$ is an invariant function on $Q$ (say, $F \circ \Phi_g = F$), it can be reduced to a function $f : Q/G \to \r$ with $f \circ \pi_Q = F$. We also have that
\begin{equation}\label{wf}
W\left(F\right) = W (f \circ \pi_Q) = \breve{W}(f) \circ \pi_Q ,
\end{equation}
that is, $\breve{W}(f)$ is the reduced function on $Q/G$ of the invariant function $W(F)$ on $Q$.

From these relations, we have the following:
\begin{prop}
If $\phi$ is an integral curve of an invariant vector field $W \in \mathfrak{X}(Q)$, then $\breve{\phi} = \pi_Q \circ \phi$ is an integral curve of the reduced vector field $\breve{W} \in \mathfrak{X}(Q/G)$.
\end{prop}

We denote by $ \xi_Q$ the \textit{fundamental vector field} on $Q$ associated to an element $ \xi$ of the Lie algebra $\mathfrak{g}$ of $G$,
$$ \xi_Q(q) = \left(\Phi_q\right)_* (e) (  \xi) .$$
We have,
\begin{itemize}
	\item A function $F \in C^{\infty}(Q)$ is invariant if, and only if, $ \xi_Q(F) = 0$ for all $ \xi  \in \mathfrak{g}$.
	\item A vector field $W \in \mathfrak{X}(Q)$ is invariant if, and only if, $[W, \xi_Q] = 0$  for all $ \xi  \in \mathfrak{g}$.
\end{itemize}

\begin{remark}{\rm }
The Lie bracket $[X,Y]$ of two $G$-invariants vector fields $X$ and $Y$ on $Q$ is also $G$-invariant.
\end{remark}

From now we will use local coordinates on $Q$ defined as follows. Let $U\subset Q/G$ be an open set over which $Q$ is locally trivial, so that
$(\pi_Q)^{-1}(U)\simeq
U\times G\, .$

 We will use coordinates $(q^\alpha)=(q^i,q^a)$ on a suitable open subset $(\pi_Q)^{-1}(U)$  (containing $U\times \, \{e\}$, where $e$ is the neutral element of the group $G$) such that $(q^i)$ are coordinates on $U$, and $(q^a)$  are coordinates on the fibre $G$. 
 
 Then, the local expression of the projection $\pi_Q: Q \to Q/G$ is:
\begin{equation}\label{TriviCoord}
	\begin{array}{ccc}
		(\pi_Q)^{-1}(U) \simeq U\times G & \longrightarrow & U \\
		\noalign{\medskip} (q^\alpha)=(q^i,q^a) & \longmapsto & (q^i),
	\end{array}
\end{equation}
where $1\leq a\leq \mathop{dim}G=d$ , \, $1\leq i\leq \mathop{dim}Q-\mathop{dim}G=n-d$.

The left action of $G$ onto $(\pi_Q)^{-1}(U)\simeq U\times G$ is given
by
\[
\Phi_g(q,h)  =  (q,gh),
\]
so, in these coordinates, we have
$$\Phi_g(q^i,q^a)=\left(q^i,(\Phi_g)^b(q^i,q^a)\right) .$$

Let
\begin{equation}\label{locX}
	X  = X^i   \fpd{}{q^i} + {\tilde X}^a   \fpd{}{q^a}
\end{equation}
be the local expression of a vector  field $X \in \mathfrak{X}(Q)$.
If $ X $ is $G$-invariant then
%$$(\Phi_g)_*(q)\left(X(q)\right)=X\left(\Phi_g(q)\right),$$
%so,
%$$
%\begin{array}{ccl}
%(\Phi_g)_*(q)(X(q))&=&X^i (q)\derpar{}{q^i}\Big\vert_{gq}+ {\tilde X}^a (q)\derpar{}{q^a}(q^b \circ \Phi_g) \derpar{}{q^b }\Big\vert_{gq}\, , \\ \noalign{\medskip}
%X(\Phi_g(q))&=&X^i (gq)\derpar{}{q^i}\Big\vert_{gq}+{\tilde X}^a (gq)\derpar{}{q^a}\Big\vert_{gq}\, ,
%\end{array}
%$$
we can conclude that the functions $X^i$ on $Q$ are invariant and they can be identified with functions on $Q/G$. 
The reduced  vector  field $\breve X  $ we had defined in (\ref{defred}) is given by
${\breve X}  = {\breve X}^i   \fpd{}{q^i}$,
where ${\breve X}^i \circ \pi_Q=X^i.$

%Se  ademais $X$ é $G$-invariante entón
%$$(\phi_g)_*(q)\left(X(q)\right)=X\left(\phi_g(q)\right),\qquad q\in Q.$$

%Calculando ambos termos desta igualdade obtéñense as siguientes expresións
%\medskip

%$$
%\begin{array}{ccl}
%	(\phi_g)_*(q)(X(q))&=&X^i (q)\derpar{}{q^i}\Big\vert_{gq}+ {\tilde X}^a (q)\derpar{}{q^a}(q^b \circ \Phi_g) \derpar{}{q^b }\Big\vert_{gq}\, , \\ \noalign{\medskip}
%	X(\phi_g(q))&=&X^i (gq)\derpar{}{q^i}\Big\vert_{gq}+{\tilde X}^a (gq)\derpar{}{q^a}\Big\vert_{gq}\, ,
%\end{array}
%$$
%das que se deduce que $$X^i (q)=X^i (gq) ,$$ é dicir, as funcións $X^i$ son funcións invariantes en $Q$ e polo tanto, definen as funcións reducidas ${\breve X}^i $ en $Q/G$.

%Así, o campo de vectores reducido $\breve X \in \mathfrak{X}(Q/G)$ ten a seguinte expresión local
%$${\breve X} = {\breve X}^i  \ds\frac{\partial}{\partial q^i} ,$$
%onde ${\breve X}^i \circ \pi_Q=X^i $. 
   \subsection{ Connections on  $\pi_Q: Q \to Q/G$ and G-invariant vector fields} \label{se52n}

Suppose we are given a principal connection on the principal bundle $\pi_Q: Q \to Q/G$. We will consider  three sets of vector fields $\{X_i\}$, 
$\{\widetilde{E}_a\}$ and $\{\widehat{E}_a\}$ on $Q$. 

The first set, $\{X_i\}$, is given by the horizontal lifts of a coordinate basis of
vector fields ${\partial}/{\partial q^i}$ on $Q/G$ by the given principal connection. These
vector fields are $G$-invariant by construction, and they form a basis of the
horizontal subspace at any point. 

The other two sets of vector fields, $\{\widetilde{E}_a\}$ and $\{\widehat{E}_a\}$, will both form a basis for the vertical space of $\pi_Q$ at each point.

The vector fields $\{\widetilde{E}_a=(E_a)_Q\}$  are the fundamental vector fields on
$Q$, associated to a basis $\{E_a\}$ of the Lie algebra
$\mathfrak{g}$. They are in general not invariant vector fields. Since they are vertical by construction, we can write
\begin{equation}\label{K}
\widetilde{E}_a=K^b_a \derpar{}{q^b},
\end{equation}
for some non-singular matrix-valued function $(K^b_a)$. 

The quasi-velocities  
corresponding to the basis $X_i,\widetilde{E}_a$ will be denoted by $(v^i,v^a)$, so for every tangent vector $v_m$, $v_m=v^i\,X_i+v^a \,\widetilde{E}_a $.

The vector fields  $\widehat{E}_a$ in the last set are defined as
\begin{equation}\label{Ehat}
\widehat{E}_a\left(q,g\right) =\left({Ad_{g} \,E_a}\right)_Q\left(q,g\right)
=\widetilde{Ad_{g} \,E_a}(q,g)
\end{equation}
where we are using the local trivialization (\ref{TriviCoord}), and $Ad$ is the adjoint map and the notation $ \xi_Q$ refers again to the fundamental vector field of $ \xi \in\g$. One easily verifies that these vector fields are all invariant.

The relation between $\widehat{E}_a$ and $ \widetilde{E}_a$ can be
expressed as
\begin{equation} \label{A}
\widehat{E}_a(q,g)\, = \, A^b_a(g)\widetilde{E}_b(q,g),
\end{equation}
where $(A^b_a(g))$ is the matrix
representing $Ad_{g}: \mathfrak{g}\to  \mathfrak{g}$ with respect to the basis $\{E_a \}$ of
$\mathfrak{g}$. In particular $A^b_a(e)=\delta^b_a$.

If we set
\begin{equation} \label{gamma}
X_i=\ds\frac{\partial }{\partial  q^i}-\gamma^b_i(q^i,q^a)\widehat{E}_b
\end{equation}
the invariance of $X_i$ implies that $\partial \gamma^b_i/\partial q^a=0$. 

The quasi-velocities  
corresponding to the basis $X_i,\hat{E}_a$ will be denoted by $(v^i,w^a)$, so for every tangent vector $v_m$, $v_m=v^i\,X_i+w^a \,\hat{E}_a $.

Let $C_{ab}^c$ be the structure constants of the Lie algebra $\mathfrak{g}$, that is $[E_a,E_b]=C_{ab}^c E_c$, and $K^a_{ij}$ be the components of the curvature of the principal connection (with respect to the vertical frame ${\widehat E}_a$),
$$
K^a_{ij}=\derpar{\gamma^a_j}{q^i}-\derpar{\gamma^a_i}{q^j}+\gamma_i^b\gamma^c_jC^a_{bc}\, .
$$ 
One easily verifies that the Lie brackets of the vector fields of interest  are as follows (see e.g.\ \cite{MC}).
\begin{lem}
\begin{equation}\label{e2}
\begin{array}{lll}
(1) \, \,   [\widetilde{E}_a,\widetilde{E}_b ]=-C_{ab}^c\widetilde{E}_c, &  (2)  \, \,[\widehat{E}_a,\widehat{E}_b ]= C_{ab}^c\widehat{E}_c, & (3) \, \,
[X_i,\widetilde{E}_a ]=0, \\[2mm] \,   (4)\, \,[X_i, \widehat{E}_a ]= \Upsilon^b_{ia} \widehat{E}_b,  &
(5)\, \,[X_i ,X_j]=-K^a_{ij} \widehat{E}_a,   & (6)\, \, [\widetilde{E}_a ,\widehat{E}_b]=0 ,
\end{array}
\end{equation}
where $\Upsilon_{ia}^b = X_i(A_a^c)
\bar{A}_c^b$.
\end{lem}
%Here, $C_{ab}^c$ are the structure constants of the Lie algebra $\mathfrak{g}$, that is $[E_a,E_b]=C_{ab}^c E_c$, and 
\begin{proof}
We have that $(3)$ and $(6)$ are consequence of that $X_i$ and  $\widehat{E}_b$ are $G$-invariant.
%
%If $\widehat{E}_a\, = \, A^b_a \widetilde{E}_b$ then 
%$$[X_i,\widehat{E}_a]=[X_i, A^b_a\widetilde{E}_a]=
%X_i(A_a^b)\widetilde{E}_a=X_i(A_a^b)
%\bar{A}_b^c 
%\widehat{E}_c=\Upsilon^c_{ia} \widehat{E}_c$$

\end{proof}

\section{Symmetry reduction of a Herglotz vector field}\label{sec5}

In this section, we show that, if the contact Lagrangian is $G$-invariant, then so is its Lagrangian SODE. The integral curves of the reduced SODE will provide the contact Lagrange-Poincar\'e-Herglotz equations.

Suppose there is an action $\Phi : G \times Q \to Q$ which is free and proper. It induces a lifted action $\Phi^{TQ} = T\Phi : G \times TQ \to TQ$ on the tangent bundle $TQ$ given by
$$\Phi^{TQ} \left(g,v_q\right) = T\Phi_g (v_q) = \left(\Phi_g\right)_*(q) (v_q),$$
where $v_q \in TQ$ and $g \in G$. The lifted action $\Phi^{TQ}$ is also free and proper and, therefore, $\pi_{TQ} : TQ \to TQ/G$ is a principal fiber bundle.

In order to give an action of $G$ on $TQ \times \r$, we will consider the following induced action
$$\Phi^{TQ\times \r}_g = \left(\Phi_g^{TQ}, Id_{\r}\right) : TQ \times \r \to TQ \times \r ,$$
so that, $G$ acts on $\r$ as the identity
$$\Phi^{TQ\times \r}\left(g,(v_q,s)\right) = (\Phi^{TQ}(g,v_q) ,s) = \left(\left(\Phi_g\right)_*(q) (v_q), s\right) $$and then the coordinate function $s$ is $G$-invariant. Therefore, the projection
$$\pi_{TQ \times \r} : TQ \times \r \to \left(TQ \times \r\right)/G \simeq \ds\frac{TQ}{G} \times \r$$ 
gives $TQ \times \r$ the structure of a principal fiber bundle over $\frac{TQ}{G} \times \r$, provided that the action $\Phi^{TQ\times \r}$ is free and proper.

%\bcr si la accion es libre y propia \enc

As before, we will denote by $ \xi_Q$ the fundamental vector field corresponding to $ \xi \in \mathfrak{g}$. From the definition of the complete lift in Section~\ref{sec2}, it follows that the fundamental vector field $ \xi _{TQ \times \r}$ given by the action $\Phi^{TQ\times \r}$ is actually the complete lift $ \xi_Q^C$ of the fundamental vector field  $ \xi_Q$ given by the action on $Q$, that is,
$$ \xi _{TQ \times \r} =  \xi_Q^C .$$

Now, from the properties in (\ref{lifts-vectors13}) and since $\partial / \partial s$ is $G$-invariant (namely, $[\partial / \partial s, \xi_Q^C] = 0$), we have the following.
\begin{lem} \label{lem51}  If the frame $\{Z_\alpha\}$ on $Q$  is invariant, then the frame $\left\{Z_\alpha^C,Z_\alpha^{V}, \frac{\partial}{\partial s}\right\}$ on $TQ \times \r$ is also invariant, with respect to the action $\Phi^{TQ\times \r}$ on $TQ \times \r$,

\end{lem}
%\begin{proof}
%The proof is a consequence of the properties in (\ref{lifts-vectors13}). So, if $\{Z_\alpha\}$ is invariant, then $[Z_\alpha, \xi_Q]=0$, and we have
%$$[Z_\alpha^C, \xi_Q^C] = [Z_\alpha, \xi_Q]^C = 0 \,,\quad [Z_\alpha^V, \xi_Q^C] = [Z_\alpha, \xi_Q]^V = 0 .$$
%Moreover $\partial / \partial s$ is also $G$-invariant, then we have that $$[\partial / \partial s, \xi_Q^C] = 0$$
%since the componentes of $\xi_Q^C$ does not depend on $s$.
%\end{proof}

Without loss of generality we may suppose that the local frame  $\{Z_\alpha\}$ on $Q$, consists of only invariant vector fields (for example, we can use the invariant frame $\{X_i,{\widehat E}_a\}$ that we had introduced in Section~\ref{se52n}).

We will use coordinates $( q^\alpha)=(q^i, q^a)$ on $Q$ that are adapted to the principal fiber bundle structure $Q\to Q/G$, as explained in Section~\ref{sec4}.

The quasi-velocities  
corresponding to the basis $X_i,\widehat{E}_a$ will be denoted by $(v^i,w^a)$, so for every tangent vector $v_q$, $v_q=v^i\,X_i+w^a \,\widehat{E}_a $, 
then the set $(q^i,q^a,v^i,w^a,s)$ represents coordinates on $TQ \times \r$ adapted to the basis 
$\{ X_i,  \widehat{E}_a, \frac{\partial}{\partial s}  \}$ of vector fields.

  {\it The coordinate $s$ is $G$-invariant, since $\xi_Q^C(s)=0$. The coordinates $q^i$ are also $G$-invariant since the fundamental vector fields $ \xi_Q$ are vertical with
respect to the projection $\pi_Q:Q \to Q/G$, and therefore we have $0= \xi_Q(q^i) = \xi_Q^C(q^i) $.

Now we will  see that the   coordinate functions $v^i,w^a$    are also $G$-invariant functions on $TQ \times \r$, in fact 
  we shall prove that
$\xi_Q^C(v^i)=0$ and $\xi_Q^C(w^a )=0.$}
%since if $\xi_{TQ\times\r}(f)=0$ then $f$ is  $G$-invariant.}

 Let $\theta$ be a  $1$-form on $Q$. We define   the function $\overrightarrow{\theta}$  on $TQ\times \r$  as
$\overrightarrow{\theta}(v_q,s) =
   \theta(v_q).$
   
In local coordinates $(q^\alpha,u^\alpha,s)$, if $\theta=\theta_\alpha \, d  q^\alpha$,
then
\begin{equation}\label{thetaarrow}\overrightarrow{\theta  } =
\theta_\alpha\, u^\alpha  .
\end{equation}
 From Remark \ref{ext}, (\ref{VClifts}) and (\ref{thetaarrow}),  we can conclude the following relations.

\begin{lem}\label{pztfle}    
Let $Z$ be a vector field on  $Q$, $f$ a function  on
$Q$, and $\theta$  a $1$-form on $Q$. Then
\begin{equation}
\label{proppztfle}
Z^C(f^V)  =  (Z(f))^V, \quad Z^{V }(f^V)= 0, \quad
Z^C(\overrightarrow{\theta  }) = \overrightarrow{({\mathcal L}_Z\theta)  }, \quad
Z^{V }(\overrightarrow{\theta  })   =      ( \theta(Z))^V.
\end{equation}
\end{lem}

%\begin{proof}
%
%Denote by  $(q^\alpha=(q^i,q^a),v^\alpha=(v^i,w^a),s)$  coordinates on $TQ \times \r$. 
%
% The first two identities are given in (\ref{VClifts}). For the others identities we have
%
%$$
%Z=Z^\alpha \derpar{}{q^\alpha}\quad \theta=\theta_\gamma dq^\gamma,
%\qquad
%L_Z\theta=\left(Z^\alpha\partial_\alpha\theta_\gamma+\theta_\alpha \partial_\gamma Z^\alpha
%\right) dq^\gamma,
%$$
%
%$$\overrightarrow{({\mathcal L}_Z\theta)}  =(L_Z\theta)_\gamma u^\gamma
%=\left(Z^\alpha\partial_\alpha\theta_\gamma+ \theta_\alpha \partial_\gamma Z^\alpha 
%\right)u^\gamma,$$
%
%$$
%Z^C(\overrightarrow{\theta  })=Z^C(\theta_\gamma u^\gamma)=Z^\alpha \partial_\alpha\theta_\gamma u^\gamma + u^\gamma \partial_\gamma Z^\alpha \theta_\alpha.
%$$
% The last identity is straightforward by using local coordinates.\end{proof}

\begin{lem}\label{pztfle1}
If  $\{Z_\alpha\}$ is a local frame on $Q$
and $\{\theta^\alpha\}$ is its dual basis, then the local quasi-velocities
 $v^\alpha$  on $TQ\times \r$ given by 
 $u^\beta = v^\alpha Z^\beta_\alpha$
  can in fact be represented by the linear functions $
v^\alpha   = \overrightarrow{(\theta^\alpha)  }$.
\end{lem}
\begin{proof} It is a straightforward computation since
$$
Z_\alpha=Z_\alpha^\beta \derpar{}{q^\beta},\qquad \theta^\alpha=\theta^\alpha_\beta dq^\beta,
\qquad \overrightarrow{\theta^\alpha  }=\theta^\alpha_\beta u^\beta=
\theta^\alpha_\beta v^\gamma Z^\beta_\gamma=v^\alpha.
$$
\end{proof}

%\begin{lem}  For   a local invariant frame on $Q$ the functions $v^\alpha=(v^i,w^a)  $ are $G$-invariant on $TQ \times \r$.
%\end{lem}
%\begin{proof}
%The fundamental vector fields $ \xi_Q$ are vertical with
%respect to the projection $\pi_Q:Q \to Q/G$, and therefore $0= \xi_Q(q^i) = \xi_Q^C(q^i) $.
Now, from Lemma~\ref{pztfle} we obtain
$$ \xi_Q^C (v^\alpha  ) =   \xi_Q^C( \overrightarrow{\theta^\alpha  })
= \overrightarrow{({\mathcal L}_{
 \xi_Q}\theta^\alpha)  }.
$$
Finally, we observe that
$$
({\mathcal L}_{  \xi_Q}\theta^\alpha)(Z_\beta ) =  {\mathcal L}_{ \xi_Q}(\theta^\alpha( Z_\beta )) - \theta^\alpha ({\mathcal L}_{  \xi_Q} Z_\beta ) = \xi_Q(\delta^\alpha_\beta ) - \theta^\alpha ([\xi_Q, Z_\beta] )      =0,
$$
since $Z_\beta $ is $G$-invariant, therefore  $\xi_Q^C (v^\alpha  )=0.$ 
%\end{proof}

In particular, if $Z_\alpha= (X_i, \widehat{E_a}, \frac{\partial}{\partial s})$ and $\theta^\alpha= (\theta^i,\theta^a,ds)$ is its dual basis, then
$   \overrightarrow{ \theta^i}=v^i , \,\,
   \overrightarrow{\theta^a}=w^a$, and we get $\xi_Q^C (v^i  )=0$ and $\xi_Q^C (w^a  )=0$.

\subsection{Invariant Herglotz SODES} \label{newsec61}

For the remainder of the paper we will assume that the Lagrangian $L$ is invariant under the action $\Phi^{TQ\times \r}$, for a connected Lie group $G$.  In view of what we said before this means that $\xi_Q^C(L)=0$,  for all $ \xi   \in \mathfrak{g}$.

\begin{prop}\label{Gammainv}
	Given $L : TQ \times \r \to \r$ a regular invariant Lagrangian, then  the \textit{Herglotz SODE} $\Gamma$    associated with the Lagrangian function $L$     is $G$-invariant, that is
	$$0 = [\widetilde{E}_a^C, \Gamma] .$$
\end{prop}

\begin{proof}
From (\ref{sodelf}) %and (\ref{elocSODE}) 
we know that if $\{Z_{\alpha}\}$ is an invariant local frame on $Q$, the local expression of a Lagrangian SODE is
\begin{equation}\label{lagsode}
\Gamma = v^\alpha Z^C_\alpha + \Gamma^\alpha Z^V_\alpha + L \ds\frac{\partial}{\partial s} .
\end{equation}
We only need to check that $[\tilde{E}_a^C, \Gamma] = 0 $, where $\widetilde{E}_a^C$ is the fundamental vector field $(E_a)_{TQ \times \r}$ of the action on $TQ \times \r$.

From the invariance of the vector fields $Z_\alpha^C$, $Z_\alpha^V$ and $\frac{\partial}{\partial s}$, and the functions $L$ and $v^\alpha$ we have
\begin{equation}\label{Corchete}
[\widetilde{E}_a^C, \Gamma] = [\widetilde{E}_a^C, v^\alpha Z^C_\alpha + \Gamma^\alpha Z^V_\alpha + L \frac{\partial}{\partial s}] = \widetilde{E}_a^C(\Gamma^\alpha) Z_\alpha^V ,
\end{equation}
so $\Gamma$ is $G$-invariant if the functions $\Gamma^\alpha$ are invariants.

We apply the vector field $\widetilde{E}_a^C$ to both sides of the equations \eqref{ecsHlocframe}
	$$\Gamma\left(Z_{\alpha}^V\left(L\right)\right) - Z_{\alpha}^C\left(L\right) = \ds\frac{\partial L}{\partial s} \left(Z_{\alpha}^V\left(L\right)\right).$$

and we obtain
$$
\begin{array}{rcl}
\widetilde{E}_a^C \left(\Gamma\left(Z_{\alpha}^V\left(L\right)\right) - Z_{\alpha}^C\left(L\right)\right) & = & \widetilde{E}_a^C \left(\Gamma\left(Z_{\alpha}^V\left(L\right)\right)\right) - \widetilde{E}_a^C \left(Z_{\alpha}^C\left(L\right)\right) \\ \noalign{\medskip}
 & = &[\widetilde{E}_a^C, \Gamma] \left(Z_\alpha^V (L)\right) + \Gamma \left(\widetilde{E}_a^C\left(Z_\alpha^V (L)\right)\right) - [\widetilde{E}_a^C,Z_\alpha^C] (L) - Z_\alpha^C\left(\widetilde{E}_a^C(L)\right) \\ \noalign{\medskip}
 & = & [\widetilde{E}_a^C, \Gamma] \left(Z_\alpha^V (L)\right),
\end{array}
$$
because of $\widetilde{E}_a^C(L) = 0$ and the zero Lie brackets of $[\widetilde{E}_a^C,Z_\alpha^C]$ and $[\widetilde{E}_a^C,Z_\alpha^V]$.

Since 
$$
0=[\widetilde{E}_a^C,\ds\frac{\partial }{\partial s} ](L)=
\widetilde{E}_a^C(\ds\frac{\partial L}{\partial s} )-
\ds\frac{\partial L}{\partial s} (\widetilde{E}_a^C(L))=\widetilde{E}_a^C(\ds\frac{\partial L}{\partial s} )
$$
	 then 
$$\widetilde{E}_a^C\left(\ds\frac{\partial L}{\partial s} \left(Z_{\alpha}^V\left(L\right)\right)\right) = \widetilde{E}_a^C\left(\ds\frac{\partial L}{\partial s}\right) Z_\alpha^V(L) + \ds\frac{\partial L}{\partial s} \, \widetilde{E}_a^C \left(Z_\alpha^V(L)\right) = 0 .$$

So, by making use of expression (\ref{Corchete}), we have
	$$0 = [\widetilde{E}_a^C, \Gamma] \left(Z_\alpha^V (L)\right) = \widetilde{E}_a^C(\Gamma^\beta) Z_\beta^V\left(Z_\alpha^V (L)\right).$$

Given that the matrix $Z_\beta^V\left(Z_\alpha^V (L)\right)$ has maximal rank for a regular Lagrangian $L$, the result follows.
\end{proof}

 Now, since $\Gamma$ is invariant, it reduces to a  vector  field $\breve{\Gamma}$ on $(TQ \times \r)/G=TQ/G \times  \r$ defined by
$$\breve{\Gamma}\left([v_q],s_0\right) = \left(\pi_{TQ \times \r}\right)_* \left(v_q,s_0\right) \Gamma \left(v_q,s_0\right) ,$$
for $v_q \in TQ$, $s_0 \in \r$.

 The goal of the next few sections is to provide a coordinate expression of this  vector  field. To do so, we will need to invoke a principal connection on the bundle $Q\to Q/G$.

\subsection{Lagrange-Poincaré-Herglotz equations} \label{sec62}

Suppose we are given a principal connection on the principal bundle $\pi_Q: Q \to Q/G$. We choose the invariant frame $\{X_i,\widehat{E}_a,\frac{\partial}{\partial s}\}$ (see Section~\ref{se52n}), and we will write $(q^i,q^a,v^i,w^a,s)$ for the coordinates and the corresponding  quasi-velocities on $TQ \times \r$.

From Lemma \ref{lem51}  we know that
the frame
 $\{Z_\alpha^C,Z_\alpha^{V  },\frac{\partial}{\partial s}\} =
\{X_i^C,X_i^{V  },\widehat{E}_a^C, \widehat{E}_a^{V  },\frac{\partial}{\partial s}\}$
 consists only of invariant vector fields on $TQ  \times \r$. Also the coordinate functions $(q^i,v^i  , w^a,s)$ are $G$-invariant functions on $TQ \times \r$ and, therefore, they can be used as coordinates on
$(TQ \times \r)/G$.
In summary, we may say that the canonical projections are locally given by
\begin{equation}\label{project}
\begin{array}{ccc}
\pi_Q: Q & \to & Q/G \\ \noalign{\medskip}
(q^i,  q^a) & \mapsto & (q^i)
\end{array}
\qquad  \quad
\begin{array}{ccc}
\pi_{TQ \times \r} : TQ \times \r & \to & (TQ \times \r)/G \\ \noalign{\medskip}
(q^i,  q^a,v^i  ,w^a,s  ) & \mapsto & (q^i,v^i  ,w^a,s  ).
\end{array}
\end{equation}
\begin{lem}\label{lem28}
If we apply the vector fields $X_i^C, X_i^{V}, \widehat{E}_a^C,
\widehat{E}_a^{V}, \frac{\partial}{\partial s}$ to the invariant functions $q^i, v^i  ,
w^a ,s $ we obtain
 \[\begin{array}{llll}
  X_i^C(q^j)=\delta^j_i\, , &  X_i^C(v^j )=0 \, , &  X_i^C(w^a )=-\Upsilon^a_{ib}w^b +K^a_{ij}v^j  \, , & X_i^C(s) = 0 \, , \\ \noalign{\medskip}
 X_i^{V  }(q^j)=0 \, , & X_i^{V  }(v^j )=\delta^j_i  \, , & X_i^{V  }(w^a)=0\, , & X_i^V(s) = 0 \,, \\ \noalign{\medskip}
 \widehat{E}_a^C(q^j)=0\, , & \widehat{E}_a^C(v^j  )=0\, , & \widehat{E}_a^C(w^b )=\Upsilon^b_{ia}v^i  -C^b_{ac}w^c  \, , & \widehat{E}_a^C(s) = 0 \, , \\ \noalign{\medskip}
 \widehat{E}_a^{V   }(q^j)=0\, , & \widehat{E}_a^{V   }(v^j  )=0\, , & \widehat{E}_a^{V   }(w^b )=\delta^b_a  \, , & \widehat{E}_a^V(s) = 0 \, , \\ \noalign{\medskip}
 \ds\frac{\partial}{\partial s}(q^j) = 0 \, , & \ds\frac{\partial}{\partial s}(v^j) = 0 \, , & \ds\frac{\partial}{\partial s}(w^a) = 0 \, , & \ds\frac{\partial}{\partial s}(s) = 1 .
  \end{array}
  \]
\end{lem}

\begin{proof} Let $\{\vartheta^j,\vartheta^a\}$ be the dual basis of $\{X_i,{\widehat E}_a\}$. From the bracket relations (\ref{e2}), we can see that ${\mathcal L}_{X_i}\vartheta^j=0$  and that ${\mathcal L}_{X_i}\vartheta^b=-\Upsilon^b_{ic}\vartheta^c+K^b_{ik}\vartheta^k$. Therefore,
$$
X_i^C(v^j)=X_i^C(\overrightarrow{\vartheta^j})= \overrightarrow{({\mathcal L}_{X_i}\vartheta^j)}=0
$$
and, from (\ref{proppztfle}), we deduce 
$$
X_i^C(w^b)=X_i^C( \overrightarrow{\vartheta^b})=
\overrightarrow{({\mathcal L}_{X_i}\vartheta^b)}=-\Upsilon^b_{ic}w^c+K^b_{ik}v^k.
$$
Since we also have
$$
X_i^C(q^j)=X_i(q^j)=\delta^j_i,
$$
then, the first row in the Lemma follows.
The other properties follow by similar arguments.%\qed
\end{proof}

\begin{lem}\label{projvf}
The  projections of the $G$-invariant vector fields
$X_i^C,X_i^{V  },\widehat{E}_a^C, \widehat{E}_a^{V  }, \frac{\partial}{\partial s}$ onto
$(TQ \times \r)/G$ are locally given, respectively, by
\[\begin{array}{ll}
 \breve{X}_i^C  =    \ds\frac{\partial}{\partial
q^i}+(K^a_{ij}v^j-\Upsilon^a_{ib}w^b)\ds\frac{\partial}
{\partial w^a  },   & \qquad
 \breve{X}_i^{V  } =  \ds\frac{\partial}{\partial v^i  },
   \\
\breve{{E}}_a^{C}  =
(\Upsilon^b_{ia}v^i  -C^b_{ac}w^c) \ds\frac{\partial}{\partial
w^b },    & \qquad
\breve{{E}}_a^{V  }  =  \ds\frac{\partial}{\partial w^a  } , \\
\breve{\ds\frac{\partial}{\partial s}} = \ds\frac{\partial}{\partial s} .
\end{array}
\]
where we are using the notations 
$$\breve{{E}}_a^{C}=\breve{\widehat{E}}_a^{C}, \quad \breve{{E}}_a^{V}=\breve{\widehat{E}}_a^{V}$$
\end{lem}

%\begin{proof}
%From the expressions in Lemma~\ref{lem28} and the relation (\ref{wf}) between an invariant vector field and its reduction, we obtain:
%$$
%\breve{X}_i^C(q^j )\circ \pi_{TQ}=X^C_i(q^j\circ \pi_{TQ})=X^C_i(q^j )=\delta^i_j,
%$$
%$$
%\breve{X}_i^C(v^j  )\circ \pi_{TQ}=X^C_i(v^j  \circ \pi_{TQ})=X^C_i(v^j  )=0,
%$$
%$$
%\breve{X}_i^C(w^j  )\circ \pi_{TQ}=X^C_i(w^j  \circ \pi_{TQ})=X^C_i(w^j  )= K^j_{ik}v^k  -\Upsilon^j_{ic} w^c .
%$$
%Since $(q^i,v^i  ,w^i  )$ forms a set of coordinate functions on $(TQ)/G$, this determines the vector field completely.
%The same idea allows us to prove the other relations.%\qed
%\end{proof}

To finish, we will get the local expression of the reduced vector field $\breve{\Gamma}$ of a Herglotz SODE $\Gamma$.

Assume that an invariant Lagrangian $L\in\cinfty{TQ  \times \r}$ is given. Since the vector fields $$X_i^{C},\, X_i^{V  }, \, \widehat{E}_a^{C} ,\, \widehat{E}_a^{V  } ,\, \ds\frac{\partial }{\partial s}$$ are also invariant, then the functions $$X_i^{C}(L),\, X_i^{V  }(L), \, \widehat{E}_a^{C}(L) ,\, \widehat{E}_a^{V  }(L) ,\, \ds\frac{\partial L}{\partial s}$$ are invariant too. From the relation (\ref{wf}) between invariant vector fields and its reduced vector fields, we can therefore write
\begin{equation}\label{ss}
\begin{array}{ll}
X_i^{V  }(L)=\breve{X}^{V  }_i(l)\circ  \pi_{TQ \times \r}, & \qquad X_i^{C}(L)=\breve{X}^{C}_i(l)\circ \pi_{TQ \times \r}, \\ \noalign{\medskip}
\widehat{E}_a^{V  }(L)=\breve{E}_a^{V  }(l)\circ
  \pi_{TQ \times \r}, & \qquad \widehat{E}_a^{C}(L)=\breve{E}_a^{C}(l)\circ \pi_{TQ \times \r}, \\ \noalign{\medskip}
  \ds\frac{\partial L}{\partial s} = \ds\frac{\partial l}{\partial s} \circ \pi_{TQ \times \r} , &
\end{array}
\end{equation}
where $l:(TQ \times \r)/G \to \r$ is the {\sl reduced Lagrangian},  defined by $l\circ
\pi_{TQ \times \r} =L$.

From Proposition \ref{frameprop}, we know that a Herglotz SODE  $\Gamma$ satisfies, with respect to the frame $\{X_i,{\widehat E}_a, \frac{\partial}{\partial s}\}$, the equations
\begin{equation}\label{eqsodeinvfr}
\begin{array}{c}
\Gamma (X_i^{V }(L))- X_i^C(L)= \ds\frac{\partial L}{\partial s} X_i^V(L),  \\ \noalign{\medskip}
\Gamma (\widehat{E}_a ^{V }(L))-\widehat{E}_a ^C(L)=\ds\frac{\partial L}{\partial s} \widehat{E}_a^V(L).
\end{array}
\end{equation}
By making use of the fact that  $\Gamma  $ is an invariant vector field on $TQ \times \r$, it follows from (\ref{wf}), (\ref{ss}) and (\ref{eqsodeinvfr})  that the reduced vector field ${\breve\Gamma}  $ satisfy
\[ \begin{array}{c}
\breve{\Gamma} (\breve{X}_i^{V }(l))- \breve{X}_i^C(l)= \ds\frac{\partial l}{\partial s} \breve{X}_i^V(l),  \\ \noalign{\medskip}
\breve{\Gamma} (\breve{E}_a ^{V }(l))-\breve{E}_a^C(l)=\ds\frac{\partial l}{\partial s} \breve{E}_a^V(l).
\end{array} \]
 on $(TQ\times \r)/G$. Taking into account the result in Lemma
\ref{projvf} we can rewrite these equations as
\begin{equation}\label{eqelred}
\begin{array}{l}
\breve{\Gamma} \left(\ds\frac{\partial l}{\partial v^i}\right) - \ds\frac{\partial l}{\partial q^i}  = \left(K^a_{ij}v^j-\Upsilon^a_{ib}w^b\right) \ds\frac{\partial l} {\partial w^a } + \ds\frac{\partial l}{\partial s} \ds\frac{\partial l}{\partial v^i},  \\ \noalign{\medskip}
\breve{\Gamma} \left(\ds\frac{\partial l}{\partial w^a}\right) = \left(\Upsilon^b_{ia}v^i  -C^b_{ac}w^c\right) \ds\frac{\partial l}{\partial w^b} + \ds\frac{\partial l}{\partial s} \ds\frac{\partial l}{\partial w^a } .
\end{array}
\end{equation}

From (\ref{lagsode}) and in terms of the frame $\{X_i,{\widehat E}_a, \frac{\partial}{\partial s}\}$, the Herglotz SODE can be written as
\begin{equation}\label{gamlocl}
\Gamma = v^i X_i^C + w^a \widehat{E}_a^C + \Gamma^i X^V_i + \Gamma^a \widehat{E}^V_a + L \ds\frac{\partial}{\partial s} .
\end{equation}
We have already established in the proof of Prop~\ref{Gammainv} that the functions $\Gamma^i $ and $\Gamma^a $
are  invariant, so they may be considered as functions on $(TQ)/G \times \r$.

From Lemma~\ref{projvf}, and since $C^c_{ab}w^aw^b=0$ and  $K^a_{ij}v^iv^j=0$, we see that:

\begin{lem}\label{ivfgam}

  The reduced vector field $\breve{\Gamma}   $ on $(TQ/G )\times \r$ of
  a Herglotz SODE $\Gamma  $ is given by
$$
\breve{\Gamma} = v^i \, \ds\frac{\partial }{\partial q^i} + \Gamma^j \,\ds\frac{\partial}{\partial v^j} +\Gamma^a \,\ds\frac{\partial}{\partial w^a}  + l \, \ds\frac{\partial }{\partial s} .
$$
\end{lem}
Now, if  $$t \longrightarrow\breve{\phi}(t) = \left(q^i(t) , v^i(t) , w^a(t) ,s(t)\right)
\in (TQ/G )\times \r$$ is an integral curve of $\breve{\Gamma}$ then it satisfies,
in view of relations (\ref{eqelred}), the \textit{Lagrange-Poincaré-Herglotz equations},
\begin{equation}\label{l-eq}
\begin{array}{ll}
\ds\frac{dq^i  }{dt} = v^i ,\\ \noalign{\medskip}
	\ds\frac{d}{dt} \left(\ds\frac{\partial l}{\partial v^i} \circ \breve{\phi}\right) - \ds\frac{\partial l}{\partial q^i} \circ \breve{\phi}  = \left(K^a_{ik} v^k  -\Upsilon^a_{ib} w^b\right) \ds\frac{\partial l} {\partial w^a } \circ \breve{\phi} + \ds\frac{\partial l}{\partial s} \ds\frac{\partial l}{\partial v^i} \circ \breve{\phi}, \\ \noalign{\medskip}  
\ds\frac{d}{dt} \left(\ds\frac{\partial l}{\partial w^a} \circ \breve{\phi} \right) = \left(\Upsilon^b_{ia}v^i -C_{ac}^b w^c\right) \ds\frac{\partial l}{\partial w^b} \circ \breve{\phi} + \ds\frac{\partial l}{\partial s} \ds\frac{\partial l}{\partial w^a } \circ \breve{\phi} .
 \\ \noalign{\medskip}  
 \ds\frac{ds}{dt}=l
\end{array}
\end{equation}

\begin{remark} 
In the case $Q=G$, if we choose the connection to be trivial, equations \eqref{l-eq} are now
    
$$\begin{array}{ccc}
\ds\frac{d}{dt} \left(\ds\frac{\partial l}{\partial w^a} \circ \breve{\phi} \right) =   -C_{ac}^b \phi^c  \ds\frac{\partial l}{\partial w^b} \circ \breve{\phi} + \ds\frac{\partial l}{\partial s} \ds\frac{\partial l}{\partial w^a } \circ \breve{\phi} .
 \\ \noalign{\medskip}  
 \ds\frac{ds}{dt}=l\end{array}
$$ 
    which are the Euler-Poicaré-Herglotz equations considered in  \cite{EPH}, \cite{A1}, \cite{A2}.

 \end{remark}

\subsection{Illustrative example: Dissipative Wong's equations}

Let $(Q, g_Q)$ be a Riemannian manifold on which a group $G$ acts freely and properly to the left as isometries, and we make the further stipulation that the vertical part of the metric (that is, its restriction to the fibres of $\pi : Q \to Q/G$) comes from a bi-invariant metric $\kappa
$ on $G$. Suposse that $\mathfrak{g}$ is the Lie algebra of $G$, $\omega : TQ \to \mathfrak{g}$ is a principal connection on $Q$ and $g$ is the reduced metric on $Q/G$. Then, the metric $g_Q$ has the following expression
$$g_Q (v_q, u_q) = \kappa(e) (\omega(v_q), \omega(u_q)) + g (\pi(q)) (\pi_*(q) (v_q), \pi_*(q) (u_q)) ,$$
for $v_q$, $u_q \in T_qQ$ and $\kappa(e)$ being the metric $\kappa$ at the identity element $e$ in $G$.

The fact that the symmetry group acts as isometries means that the fundamental vector fields $\widetilde{E_a}$ are Killing vector fields, $\Li_{\widetilde{E_a}} \, g_Q = 0$. It follows that the components of $g_Q$ with respect to the members of an invariant basis $\{X_i,\widehat{E_a}\}$ are themselves invariant. For this metric, we have $g_Q (X_i, \widehat{E_a}) = 0$ and we will set $g_Q(X_i,X_j) = g_{ij}$ and $g_Q(\widehat{E_a}, \widehat{E_b}) = h_{ab}$.
Since both $g_{ij}$ and $h_{ab}$ are $G$-invariant functions, they pass to the quotient; in particular $g_{ij}$ are the components with respect to the coordinate fields of the metric $g$ on $Q/G$.

The further assumption about the vertical part of the metric has the following implications.
It means in the first place that $\Li_{\widehat{E_c}} g_Q(\widehat{E_a}, \widehat{E_b}) = 0$ (since $\Li_{\widetilde{E_c}} g_Q(\widehat{E_a}, \widehat{E_b}) = 0$), and secondly that the $h_{ab}$ must be independent of the coordinates $q^i$ on $Q/G$, which is to say that they must be constants. From the first condition, taking into account the bracket relations $[\widehat{E_a}, \widehat{E_b}] = C_{ab}^c \widehat{E_c}$, we easily find that $h_{ab}$ must satisfy $h_{bd} C_{ac}^b + h_{bc} C_{ad}^b = 0$. Also, if we set
$$X_i = \ds\frac{\partial}{\partial q^i} - \gamma_i^a \widehat{E_a}$$
for some $G$-invariant coefficients $\gamma_i^a$, we get $\Upsilon_{ia}^b = \gamma_i^c C_{ac}^b$, and therefore $h_{ac} \Upsilon_{ib}^c + h_{bc} \Upsilon_{ia}^c = 0$.

We consider the kinetic contact energy $L: TQ \times \r \to \r $ associated with $g_Q$, which is given by the Lagrangian
$$L(v_q,s) = \ds\frac{1}{2} \left( \kappa(e)\left(\omega(v_q), \omega(v_q)\right) + g (\pi(q)) \left(\pi_* (q) (v_q),  \pi_* (q) (v_q)\right)\right) - \gamma s,$$
for $v_q \in T_qQ$ and $s \in \r$.
In local coordinates $(q^i,q^a, v^i, w^a, s)$, we have
$$L= \ds\frac{1}{2} \left(g_{ij} v^i v^j +  h_{ab} w^a w^b \right)- \gamma s .$$

It is clear that $L$ is hyperregular and $G$-invariant.

On the other hand, since the Riemannian metric $g_Q$ is also $G$-invariant, it induces a fiber metric $g_{TQ/G}$ on the quotient vector bundle $TQ/G \to Q/G$. The reduced contact Lagrangian $l : TQ/G \times \r \to \r$ is just the kinetic contact energy of the fiber metric $g_{TQ/G}$, that is,
$$l([v_q],s) = \ds\frac{1}{2}\left(\kappa(e)\left(\omega(v_q), \omega(v_q)\right) + g(\pi(q)) \left(\pi_* (q) (v_q),  \pi_* (q) (v_q)\right)\right) -\gamma s ,$$
which in local coordinates $(q^i,v^i,w^a,s)$ is given by
$$l = \ds\frac{1}{2} \left(g_{ij} v^i v^j +  h_{ab} w^a w^b \right)- \gamma s.$$

The Lagrange-Poincaré-Herglotz equations \eqref{l-eq} for the contact Lagrangian function $l$ are
\begin{equation}\label{LPHeq}
\begin{array}{rl}
   \ds\frac{d}{dt}\left(g_{ij} v^j \right) - \ds\frac{1}{2} \ds\frac{\partial g_{jk}}{\partial q^i} v^j v^k =  & \left( B_{ij}^a v^j - \Upsilon_{ib}^a w^b \right) h_{ac} w^c - \gamma g_{ij} v^j \,, \\ \noalign{\medskip} 
   \ds\frac{d}{dt} \left(h_{ab} w^b \right) =  & \left(\Upsilon
_{ia}^c v^i - C_{ad}^c w^d \right) h_{ce} w^e - \gamma h_{ab} w^b \,, \\ \noalign{\medskip}
\ds\frac{d s}{dt} = & l \,,
\end{array}
\end{equation}
where $B_{ij}^a$ are the components of the curvature of the principal connection $\omega$.

Since $\Upsilon_{ib}^a h_{ac}$ is skew-symmetric in $b$ and $c$, and $C_{ad}^c h_{ce}$ is skew-symmetric in $d$ and $e$, the terms $\Upsilon_{ib}^a w^b  h_{ac} w^c$ and $C_{ad}^c w^d h_{ce} w^e$ in the first and second equations vanish identically. Now, let $\Gamma_{jk}^i$ be the connection coefficients of the Levi-Civita connection of the reduced metric $g_{ij}$, then we may write the equations in the form
\begin{equation}
    \begin{array}{rl}
      g_{ij} \left(\Ddot{q}^j + \Gamma_{kl}^j \dot{q}^k \dot{q}^l \right) =   & h_{ac} B_{ij}^a \dot{q}^j w^c - \gamma g_{ij} \dot{q}^j \\ \noalign{\medskip}
      h_{ab} \left(\dot{w}^b + \Upsilon_{id}^b \dot{q}^i w^d + \gamma w^b \right) =    & 0
    \end{array}
\end{equation}
using the skew-symmetry of $\Upsilon_{ia}^b h_{bd}$ again in the second equation. Given that $B_{ij}^a$ is of course skew-symmetric in its lower indices, these equations are equivalent to
\begin{equation}\label{disipWongeq}
 \begin{array}{rl}
    \Ddot{q}^m + \Gamma_{kl}^m \dot{q}^k \dot{q}^l = & g^{im} h_{bc} B_{ji}^c \dot{q}^j w^b - \gamma \dot{q}^m \\ \noalign{\medskip}
     \dot{w}^a + \Upsilon_{ib}^a \dot{q}^i w^b + \gamma w^a = & 0
 \end{array}   
\end{equation}
 These are dissipative Wong's equations (see \cite{montgomery}) for the usual case).

Let us now take $Q$ to be $E^3 \times S$ with coordinates $(x^i,\theta)$. Let $A_i$ be the components of a covector field on $E^3$, and define a metric $g_Q$ on $Q$, the Kaluza-Klein metric, by
\begin{equation}
    g_Q = \delta_{ij} dx^i \otimes dx^j + (A_i dx^i + d\theta)^2 ,
\end{equation}
where $(\delta_{ij})$ is the Euclidean metric. The Kaluza-Klein admits the Killing field $E= \frac{\partial}{\partial \theta}$.
The vector fields $X_i = \frac{\partial}{\partial x^i} - A_i \frac{\partial}{\partial \theta}$ are orthogonal to $E$ (that is, $g_Q(X_i,E) = 0$), and invariant. Moreover, $g_{ij} = g_Q (X_i,X_j) = \delta_{ij}$, while $g_Q(E,E) = 1$. Finally,
$$[X_i,X_j] = \left(\ds\frac{\partial A_i}{\partial x^j} - \ds\frac{\partial A_j}{\partial x^i}\right) \ds\frac{\partial}{\partial \theta} .$$

Putting these values into the reduced equations above \eqref{disipWongeq} we obtain
\begin{equation}
    \Ddot{x}^m  =  \left(\ds\frac{\partial A_i}{\partial x^j} - \ds\frac{\partial A_j}{\partial x^i}\right) \dot{x}^j w - \gamma \dot{x}^m \, , \qquad \dot{w} +  \gamma w =  0 .
\end{equation}

\section{The principal connection on $ TQ \times \r\to   TQ/G \times \r$} \label{sec6}
In order for the reconstruction method to work, we need a     connection on $\pi_{   TQ \times \r } :TQ \times \r\to (TQ \times \r)/G\equiv TQ/G \times \r$. We start by defining the horizontal distribution, constructed from a given Lagrangian function satisfying a certain regularity assumption.

Consider the    contact  form $\eta_L$ associated to a Lagrangian $L$. We define linear maps
\[ \begin{array}{ccccl}
g_{(v_q,s) }   & : & T_qQ \times    T_qQ & \to & \r
\\ \noalign{\medskip}
  &  &    (u_q,w_q) & \to &    g_{(v_q,s) }  (u_q,w_q)= d\eta_L(v_q,s) (X^C(v_q,s) ,Y^{V}(v_q,s) ),
\end{array}\]
where $X,Y$ are vector fields on $Q$ for which $X(q)=u_q$ and $Y(q)=w_q$.

In the natural coordinates  $(q^\alpha  ,u^\alpha , s    )$ on $TQ\times \r  $, the coordinate expression of $g_{(v_q,s) }  $ is
\begin{equation}
 \label{localG}
g_{(v_q,s) }   =\displaystyle \frac{\partial^2 L}{\partial u^\alpha      \partial u^\beta}\Big\vert_{ (v_q,s) } dq^\alpha  (v_q,s)  \otimes dq^\beta(v_q,s) .
\end{equation}

 %Then if $X,Y$ are vector fields on $Q$ then $g(X,Y)=X^V(Y^V(L))$.
 In what follows, we will use the following notations for the coefficients with respect to the basis $\{X_i,{\tilde E}_a\}$ of vector fields on $Q$:
$$\begin{array}{lcl}
   g  _{ia}(v_q,s)  &=& g_{(v_q,s) }   (X_i(q),  {\widetilde E}_a(q)) =g_{(v_q,s) }   ( {\widetilde E}_a(q),X_i(q)) =g  _{ai}(v_q,s) , \\
\noalign{\medskip}
g  _{ij}(v_q,s)  &=& g_{(v_q,s) }   (X_i(q),X_j(q)),   \quad
g  _{ab}(v_q,s)  =g_{(v_q,s) }   (  {\widetilde E}_a(q),{\widetilde E}_b(q)),
\end{array}$$
from  the local expressions \eqref{K},\eqref{Ehat}, \eqref{gamma} and \eqref{localG} we obtain 
\begin{equation}\label{hess}
g_{ij} = {X}_i^{V   } ({X}_j^{V}(L)),\quad g_{ib} = {X}_i^{V   } ({\widetilde E}_b^{V}(L)), \quad g_{ab} =   {\widetilde E}_a^{V   } ({\widetilde E}_b^{V}(L)).
\end{equation}

 \begin{definition}
    A Lagrangian $L$ is $G$-regular if the matrix $(g_{ab})$ is non-singular.
\end{definition}

Let us observe that 
$$
g_{ab}=  {\widetilde E}_a^{V}
({\widetilde E}_b^{V}(L))=K^c_a K^d_b\left( \frac{\partial^2 L}{\partial u^c \partial u^d } \right)
$$

%{\bf  We will next define the  $1$-form   connection $\Omega: T(TQ\times \r )  \to {\mathfrak g}$.}
and thus $L$ is $G$-regular if the matrix $\left( \displaystyle\frac{\partial^2 L}{\partial u^c \partial u^d } \right)$ is non-singular.

We define now the horizontal subspace of the  connection on
$\pi_{   TQ \times \r } :TQ \times \r\to (TQ \times \r)/G\equiv TQ/G \times \r$ as follows.

\begin{definition}
 An element  $W_{(v_q,s)}   \in   T (TQ\times \r  )$  is said to be horizontal   if it satisfies
\[
g_{(v_q,s) }\left(    (\tau_{Q})_*(v_q,s)W_{(v_q,s)}, \xi_Q(q)\right) =0,
\]
for all   $ \xi   \in {\mathfrak g}  $, and where $(\tau_{Q})_*(v_q,s):T_{(v_q,s)}(TQ\times \r)\to T_qQ$.
\end{definition}

Of course, this defines a splitting only if we suppose that $g$ is non-singular when restricted
to the set of fundamental vector fields; this will be the case in particular if $g$ is everywhere
positive definite.

Each element $W_{(v_q,s)}$ can be written with respect to the lifted frame 
$\{X_i^C,  {\widetilde E}^C_a,
X_i^V,  {\widetilde E}^V_a,
\frac{\partial}{\partial s}\}$
   as follows 
\begin{equation}\label{general}
W_{(v_q,s)} = W^i    X^C_i (v_q,s)  + W^a      {\widetilde E}^C_a (v_q,s)  + Z^i     X^{V}_i (v_q,s)  + Z^a       {\widetilde E}^{V}_a (v_q,s) +f\, \derpar{}{s}\Big\vert_{(v_q,s)}\, .
\end{equation}
and, then, the condition for $W_{(v_q,s)} $ to be horizontal becomes
\begin{equation}\label{general0}
g_{ib}W^i    + g_{ab}W^a     =0.
\end{equation}
Since we assume that the Lagrangian is regular, we can conclude that a horizontal vector $W_{(v_q,s)}$ takes the form
\[
W_{(v_q,s)} = W^i     H_i  (v_q,s)  + Z^i     X^{V}_i (v_q,s)  + Z^a       {\widetilde E}^{V}_a (v_q,s)+f\, \derpar{}{s}\Big\vert_{(v_q,s)}
 ,
\]
where 
$$    H_i  =      X^C_i -  B^a_{ i} {\widetilde E}^C_a$$
with ${  B}^{  a}_{  i} = g^{b a}  g_{ib}$.

Then we have  that every element of $  T ( TQ\times \r)$ can be written in a `horizontal part' ${\bf H}W$ and a `vertical part' ${\bf V}W$. Indeed,  from (\ref{general}),
 we write as follows
  $$W _{(v_q,s)} = hor\, W    + vert \, W   $$ with
$$ \begin{array}{ccl}
    {\bf H}W=hor \, W & = &    W^i   \,  H_i  (v_q,s)    + Z^i  \,   X^{V}_i (v_q,s)  + Z^a    \,   {\widetilde E}^{V}_a (v_q,s)+f\, \derpar{}{s}\Big\vert_{(v_q,s)}
    \\ \noalign{\medskip}
   {\bf V}W= vert \, W  & = & (W^a       + W^i  { B}^{  a}_{  i} ) \,{\widetilde E}^C_a (v_q,s) 
\end{array}$$
      
 The corresponding connection map $\Omega  :  
  T(TQ\times \r    ) \to {\mathfrak g}  $
  is the one that has the property that
\[
\Omega   (hor \, W) =0, \qquad \Omega   (\xi _{   TQ \times \r}(v_q,s)  ) = \xi .
\]
and thus we obtain the following result.

\begin{prop}\label{omegavectors}
\begin{equation}\label{00lift}
    \Omega( X_i^C)=B^a_i\,  E_a  ,\quad
 \Omega( \hat{ E}^{C}_a)=
 A^b_a\,E_b,\quad 
 \Omega( X_i^V)=\Omega(\widetilde{ E}^{V}_a )=0, \quad 
    \Omega(\widetilde{ E}^{C}_a)=E_a, \quad\Omega(  \derpar{}{s})=0. 
\end{equation}
\end{prop}

\begin{proof}
    
We only need to prove the following

$$0=\Omega( H_i)=\Omega(  X^C_i -  B^a_{i} \widetilde {E}^C_a)= \Omega(  X^C_i)-\Omega(  \widetilde {E}^C_a)= \Omega(  X^C_i)-B^a_i\,  E_a$$

$$ \Omega( \hat E_a^C)= (A^b_a \widetilde{E}_a)^C=
 \Omega \left((A^b_a)^C \widetilde{E}_a^V+
 A^b_a\, \widetilde{E}_a^C\right)
 =A^b_a \, {E}_b$$
\end{proof}

%\begin{equation}\label{00gamlocl}
%\Gamma = v^i X_i^C + w^a \widehat{E}_a^C + \Gamma^i X^V_i + \Gamma^a \widehat{E}^V_a + L \ds\frac{\partial}{\partial s} .
%\end{equation}

\begin{prop} Let $L$ be a $G$-regular and $G$-invariant Lagrangian. Then  
 $\Omega$ is  a principal    connection on the principal bundle $\pi_{   TQ \times \r } :TQ \times \r\to (TQ \times \r)/G\equiv TQ/G \times \r$.
\end{prop}

\begin{proof}
The condition we need to check is \begin{equation}\label{important}
    {\mathcal L}_{\xi_{   TQ\times \r }} \bar\Omega  =
{\mathcal L}_{\xi_{   TQ }} \bar\Omega
=
{\mathcal L}_{\xi_{Q}^C} \bar\Omega   = 0 ,
\end{equation}
where $\bar \Omega  $ is the connection $(1,1)$-tensor field that is associated to $\Omega  $ and ${\mathcal L}$ is the Lie derivative (see Section 2).

This tensor field is given by
\begin{equation}\label{hv0}
 \bar\Omega({\bf H} \, W)=0,\qquad \bar\Omega({\bf V}W)={\bf V}W
\end{equation}

Let $W$ be a vector field on $ TQ\times \r$, then
 $$
 \begin{array}{ccl}
   {\mathcal L}_{\widetilde{E}_a ^C} \bar\Omega (W)
   & =&  [\widetilde{E}_a ^C, \bar\Omega( W)]-
   \bar\Omega([\widetilde{E}_a ^C,W])\\ \noalign{\medskip}
   &=&  -\bar\Omega([\widetilde{E}_a ^C,{\bf H}W])= -{\bf V}[\widetilde{E}_a ^C,{\bf H}W]
   \end{array}$$
 where have used \eqref{hv0} and that $[\widetilde{E}_a ^C,{\bf V}W]$ is vertical. Computing $[\widetilde{E}_a ^C,{\bf H}W]$ we obtain that its vertical part is 
 \begin{equation}
     {\bf V}[\widetilde{E}_a ^C,{\bf H}W]= W^i[\widetilde{E}_a ^C,H_i]
 \end{equation}
 and thus ${\mathcal L}_{\widetilde{E}_a ^C} \bar\Omega (W)=0$ if the vector fields $H_i$ is invariant.

 A direct computation shows that
 \begin{equation}
     {\bf V}[\widetilde{E}_a ^C,H_i]=  \left(  -   
     \widetilde{E}_a ^C(B^d_i)+B^b_iC^d_{ab}\right)
     \widetilde{E}_d ^C.
 \end{equation}
therefore $H_i$ is invariant if \begin{equation}\label{fin}
    \widetilde{E}_a ^C(B^d_i)=B^b_iC^d_{ab}.
\end{equation}

Now, using the identities \eqref{hess} and that $L$ is invariant one obtains
\begin{equation}
  \widetilde{E}_d ^C(g_{ab})=
  C^e_{da}g_{eb}+C^e_{db}g_{ae},
  \quad \widetilde{E}_d ^C(g_{ib})=
  C^e_{db}g_{ie} ,\quad
  \widetilde{E}_d ^C(g^{ec})=-
  \widetilde{E}_d ^C(g_{ab}) g^{ac}g^{eb},
\end{equation}
and with these identities one deduce finally \eqref{fin} .
 \end{proof}

\section{Reconstruction} \label{sec7}

Now that we have derived the reduced form of the  Herglotz-Euler-Lagrange equations it remains to
consider the problem of reconstruction: suppose we can find a solution of these equations, that is
an integral curve  $\breve{\Gamma}$, how do we reconstruct from it a solution of the original equations, that is, an integral curve of $\Gamma$.

There is in fact a standard method for reconstructing integral curves of an invariant vector
field from reduced data, which makes use of connection theory \cite{AM,book}.

   \bigskip

   \bigskip

   Let 
 $\breve{c}(t) $ be 
an integral curve of  $\breve{\Gamma}$  on $TQ/G\times \r$ and let $(v_{q_0},s_o)$ be a point of $TQ\times \r$ in the fibre over $\breve{c}(0)$: we aim to find
the integral curve of $\Gamma$ though $(v_{q_0},s_o)$.

Let $\breve{c}^H(t)$ be the horizontal lift through $(v_{q_0},s_o)$ of $\breve{c}(t)$ with respect to the connection $\Gamma$ on $TQ\times \r$, this is the unique curve in $TQ\times \r$ projecting onto $\breve{c}(t) $ such that
$\Omega(\dot{\breve{c}}^H(t))=0 $  and $\breve{c}^H(0) = (v_{q_0},s_o)$.
 
  Since $$\pi^{TQ\times \r}(c(t))=\pi^{TQ\times \r}(\breve{c}^H(t))=\breve{c}(t)$$
  there exists a curve $g(t)$ on $G$ such that 
\begin{equation}\label{recon1}  
    c(t)=\Phi^{TQ\times \r}(g(t),\breve{c}^H(t))\equiv 
 g(t) \, \cdot \, \breve{c}^H(t)
\end{equation}
 with $g(0)=e$.
 
Now, we define a curve $\xi(t)$ on $\mathfrak{g}$ as follows 
\begin{equation}\label{0recon1}   \xi(t)=(L_{g(t)^{-1}})_*(g(t))\dot{g}(t) \end{equation}
or, equivalently, 
  $$
\dot{g}(t)=(L_{g(t)})_*(e)\xi(t) \qquad \mbox{with} \quad g(0)=e.
$$ 

Now, using the following lemma 
\begin{lem}
For an action $\Phi:G \times M \to M$  we have
$$
(\Phi_m)_*(g)v_g=(\Phi_g)_*(m)\eta_M(m)
\quad  \eta= (L_{g^{-1}})_*(g)v_g\in \mathfrak{g}$$
\end{lem}
 and \eqref{recon1} and \eqref{0recon1}, we get
 \begin{equation}\label{recn2}
 \begin{array}{ccl}
   \Gamma(c(t))  =\dot{c}(t) &=&
   (\Phi^{TQ\times \r}_{g(t)})_*(\breve{c}^H(t))\,  \dot{\breve{c}}^H(t) \, + \, ((\Phi^{TQ\times \r})_{\breve{c}^H(t)})_*(g(t)) \dot{g}(t)
\\ \noalign{\bigskip}
&=&    (\Phi^{TQ\times \r}_{g(t)})_*(\breve{c}^H(t)) \left[ \dot{\breve{c}}^H(t)\, + \,    (\xi(t))_{TQ\times \r}(\breve{c}^H(t)) \right]  \in T_{c(t)}(TQ\times \r )
\end{array}
 \end{equation}

From equation \eqref{recn2} and since $\Gamma$ is invariant, we deduce
 $$ 
 \dot{\breve{c}}^H(t)+    (\xi(t))_{TQ\times \r}(\breve{c}^H(t))  
=  \Gamma(\breve{c}^H(t))  
 $$ 
then, the reconstruction equation is
\begin{equation}\label{reconseq}      
  \Omega(\breve{c}^H(t))( \Gamma(\breve{c}^H(t))) = \xi(t) \, .
\end{equation}

 From (\ref{gamlocl}) and Proposition \ref{omegavectors} we have
$$ 
     \Omega( X_i^C)=B^a_i\,  E_a  ,\quad
 \Omega( \hat{ E}^{C}_a)=
 A^b_a\,E_b,\quad 
 \Omega( X_i^V)=\Omega(\widetilde{ E}^{V}_a )=0, \quad 
    \Omega(\widetilde{ E}^{C}_a)=E_a, \quad\Omega(  \derpar{}{s})=0. 
$$ 

and therefore  
\begin{equation}\label{gamlocl0}
\Omega(\Gamma)= (A^a_bw^b +  B^a_i v^i)E_a\, .
\end{equation}

The reconstruction equation is therefore
\begin{equation}\label{gamlocl01}
\begin{array}{rl} 
\xi(t)=&\left[A^a_b(\breve{c}^H(t))w^b(\breve{c}^H(t)) +  B^a_i(\breve{c}^H(t)) v^i(\breve{c}^H(t))\right]E_a\, \\ \noalign{\medskip}
= & \left[A^a_b(\breve{c}^H(t))w^b( t) +  B^a_i(\breve{c}^H(t)) v^i(t)\right]E_a\,.
\end{array}
\end{equation}

since $v^i$ and $w^a$ are invariant, and then $v^i(t)$ and $w^a(t)$ are just their values on $\breve{c}(t)$.
\bigskip 

 Let us denote by   
$$  \breve c(t)=(q^i(t), v^i(t),w^a(t),s(t))
 \qquad \breve{c}^H(t)=(q^i(t),h^a(t),v^i(t),w^a(t),s(t))$$ 
the local coordinates of the corresponding of curves, where $h(t)=(h^a(t))$ is a curve on $G$ determined by $(v_{q_0},s_o)$ and the relation $\Omega(\dot{\breve{c}}^H(t))=0$.

The reconstruction equation \eqref{reconseq} is therefore
$$
\xi(t)=(A^a_bw^b(\breve{c}^H(t))+B^a_i v^i(\breve{c}^H(t)))E_a
=(A^a_bw^b(t)+ B^a_i v^i(t))E_a    
$$

%In fact in the
%coordinate system we used in the previous section, corresponding to a chart $U\times  G$ on $Q$, the
%(left) action on $Q$ is simply given by \Phi^Q_g
 %(x, h) = (x, gh)$ and the induced action on $TQ\times \r$ by
%$\Phi^{TQ\times \r}(q,h,v^i,w^a)=(q,gh,v^i,w^a).$
  
We assume that we are able to calculate the integral curve
 
$$\breve{c}(t)=(q^i(t),v^i(t),w^a(t),s(t))$$ 

of the reduced vector fiel $\breve{\Gamma}$ through $\pi^{TQ\times\r}(v_{q_0},s_0)$.

  The horizontal lift of
$ \breve{c}$ is a curve in $TQ\times \r$  of the form 
 $$\breve{c}^H(t)=(q(t),h(t),v^i(t),w^a(t),s(t))$$
  where $h(t)$ is a curve in   $G$ to be determined by 
  $(v_0,s_0)$ and by the relation
%  $(v_0,s_0)$ and by the relation 
  $\Omega(\dot{\breve{c}}^H) = 0$.

%\end{document}
The reconstruction  equation can therefore be written as
$$\xi(t)
=\left(A^a_b(h(t))\, w^b(t) +  B^a_i (  \breve{c}^H(t)   ) \, v^i(t) \right)E_a$$
and the integral curve of 
  $\Gamma$ is just 
$$c(t)=(q(t ), g(t)h(t), v^i (t),w^a(t)).$$

\subsection{A full example}

Let $G$ be the Lie group of invertible affine transformations on the real line. An element of this group is an affine map of the form \(
T_{(\theta, \phi)}:\mathbb{R} \to \mathbb{R}, \  t \mapsto \exp(\theta)t + \phi
\)
and can be represented by the matrix
\[
\begin{pmatrix}
\exp(\theta) & \phi \\
0 & 1
\end{pmatrix}.
\]
The identity element is just \( t \mapsto t \) (the identity matrix) and multiplication on the left of \((\theta_2, \phi_2)\) by \((\theta_1, \phi_1)\) is given by the composition of the two affine maps, i.e. the element
\[
(\theta_1, \phi_1) \ast (\theta_2, \phi_2) = (\theta_1 + \theta_2, \exp(\theta_1)\phi_2 + \phi_1).
\]
The corresponding Lie algebra is the vector space spanned by the matrices
\[
E_{1} =\begin{pmatrix}
1 & 0 \\
0 & 0
\end{pmatrix} \text{ and } E_{2} =\begin{pmatrix}
0 & 1 \\
0 & 0
\end{pmatrix}.
\]
Take now the manifold \( Q \) to be \( G \times \mathbb{R} \) and the action of $G$ on $Q$ given by the left translation on the first factor of \( Q \). We will denote by \( x \) the coordinate on \( \mathbb{R} \). Then $\pi:Q \rightarrow \R$ is a (trivial) principal fibre bundle and there is a trivial principal connection associated to the identification of $TQ$ with the direct sum $TG\oplus T\R$.

A basis of fundamental vector fields is
\[
\tilde{E}_1 = \frac{\partial}{\partial \theta} + \phi \frac{\partial}{\partial \phi}, \quad \tilde{E}_2 = \frac{\partial}{\partial \phi},
\]
which span the vertical space $V=\ker T\pi=TG$. In addition, the horizontal space with respect to the trivial connection is spanned by the vector field \( X = \frac{\partial}{\partial x} \). Equivalently, the connection form is the map $\omega:TQ\rightarrow \mathfrak{g}$ given by $$\omega(\gamma^{0}X + \gamma^{1}\tilde{E}_1 + \gamma^{2}\tilde{E}_2)=\gamma^{1}E_{1} + \gamma^{2}E_{2}.$$

The adapted coordinates \((v^i, v^a)\) with respect to this basis are \( v^0 = \dot{x} \) and \( v^1 = \dot{\theta}, v^2 = \dot{\phi} - \phi \dot{\theta} \).
The only non-vanishing Lie algebra structure cosntant is $C_{12}^{2}=-1$ in virtue of the fact that \([\tilde{E}_1, \tilde{E}_2] = -\tilde{E}_2\). The complete and vertical lifts of this basis are
\[
\tilde{E}^C_1 = \frac{\partial}{\partial \theta} + \phi \frac{\partial}{\partial \phi} + \dot{\phi} \frac{\partial}{\partial \dot{\phi}}, \quad \tilde{E}^C_2 = \frac{\partial}{\partial \phi}, \quad \tilde{E}^V_1 = \frac{\partial}{\partial \dot{\theta}} + \phi \frac{\partial}{\partial \dot{\phi}}, \quad \tilde{E}^V_2 = \frac{\partial}{\partial \dot{\phi}}.
\]
Alternatively, an invariant basis of vector fields is given by \(\{ \hat{E}_1, \hat{E}_2, X\}\), where
\[
\hat{E}_1 = \frac{\partial}{\partial \theta}, \quad \hat{E}_2 = \exp(\theta) \frac{\partial}{\partial \phi},
\]
and its adapted coordinates are \( v^0 = \dot{x}, w^1 = \dot{\theta} \) and \( w^2 = \exp(-\theta) \dot{\phi} \). The complete and vertical lifts of the above basis are
\[
X^C = \frac{\partial}{\partial x}, \quad X^V = \frac{\partial}{\partial \dot{x}}, \quad \hat{E}^C_1 = \frac{\partial}{\partial \theta}, \quad \hat{E}^C_2 = \exp(\theta) \left( \frac{\partial}{\partial \phi} + \dot{\theta} \frac{\partial}{\partial \dot{\phi}} \right),
\]
\[
\hat{E}^V_1 = \frac{\partial}{\partial \dot{\theta}}, \quad \hat{E}^V_2 = \exp(\theta) \frac{\partial}{\partial \dot{\phi}}.
\]
In addition, the matrix \( A \) of change of basis \(\hat{E}_a(x, g) = A^b_a(g) \tilde{E}_b(x, g)\), is
\[
A(g) =
\begin{pmatrix}
1 & 0 \\
-\phi & \exp(\theta)
\end{pmatrix}.
\]

If we use the invariant fiber coordinates \((v^0, w^a)\), the induced action on \(TQ\) is
\[
\psi_{TM}(\theta_1, \phi_1)(x, (\theta, \phi), \dot{x}, w^1, w^2) = (x, (\theta_1, \phi_1) \ast (\theta, \phi), \dot{x}, w^1, e^{\theta_{1}}w^2).
\]
Since the coordinates \((x, \dot{x}, w^1, w^2)\) can be interpreted as coordinates on \(TQ/G = T\mathbb{R} \times TG/G = T\mathbb{R} \times \mathfrak{g}\), $G$-invariance of a Lagrangian function implies that the group variables $h=(\theta,\phi)$ do not explicitly appear in the Lagrangian, when it is written in terms of the invariant fiber coordinates.

Consider the contact Lagrangian $L:TQ \times \R \to \R$ given by
\[
L(x, h, \dot{x}, \dot{h}, s) = \frac{1}{2} \dot{\theta}^2 + q \dot{x} \dot{\theta} + \frac{1}{2} \dot{x}^2 + \ln(\exp(-\theta)\dot{\phi}) - \gamma s,
\]
where \( q \) is a constant. The Lagrangian is $G$-invariant since \(\tilde{E}^C_1(L) = 0 = \tilde{E}^C_2(L)\) and in the invariant fiber coordinates \((\dot{x}, w^1, w^2)\), the Lagrangian is
\[
L = \frac{1}{2}(w^1)^2 + q \dot{x} w^1 + \frac{1}{2} \dot{x}^2 + \ln(w^2) - \gamma s.
\]

The Hessian matrix in the basis \(\{ \tilde{E}_a, X \}\) is given by
\[
g =
\begin{pmatrix}
1 - \frac{\phi^2}{\dot{\phi}^2} & -\frac{\phi}{\dot{\phi}^2}  & q \\
-\frac{\phi}{\dot{\phi}^2} & -\frac{1}{\dot{\phi}^2} & 0 \\
q & 0 & 1
\end{pmatrix}.
\]

The determinant of \( g \) is \((q^2 - 1)/\dot{\phi}^2\), so the Lagrangian is regular as long as \( q^2 \neq 1 \). The upper left \( (2,2) \) matrix represents \((g_{ab})\). It is non-singular since its determinant is \(-1/\dot{\phi}^2\). Its inverse is
\[
(g^{ab}) =
\begin{pmatrix}
1 & -\phi \\
-\phi & \phi^2 - \dot{\phi}^2
\end{pmatrix}.
\]

%\textcolor{red}{The vector field \(\bar{X}\) along \(\tau\) that projects onto \(\frac{\partial}{\partial x}\) on \(M/G\) and is horizontal for the generalized mechanical connection \(\omega_m\) is
%\[
%\bar{X} =
%\frac{\partial}{\partial x}- g^{bc} g_{cx} \tilde{E}_b
%=\frac{\partial}{\partial x} - q \tilde{E}_1 + q \phi \tilde{E}_2 = \frac{\partial}{\partial x}- q\frac{\partial}{\partial \theta}.\]
%This vector field is in fact a basic vector field along \(\tau\). Although the Lagrangian and the Hessian are not of the simple type, in this example the generalized mechanical connection turns out to be derived from a principal connection on \( M \to M/G \), and is not of the most general case of an invariant connection on the pullback bundle.
%The corresponding vector field that is horizontal with respect to the connection \(\omega_m\) is\ \[\bar{X}^C =\frac{\partial}{\partial x}- q \tilde{E}^C_1 + q \phi \tilde{E}^C_2=\frac{\partial}{\partial x}- q \frac{\partial}{\partial \theta}- q \dot{\phi} \frac{\partial}{\partial \dot{\phi}}.\]}

The Euler-Lagrange-Herglotz equations for the Lagrangian $L$ are
\[
q \ddot{\theta} + \ddot{x} = -\gamma (\dot{x}+q\dot{\theta}), \quad \ddot{\theta} + q \ddot{x} + 1 = -\gamma (\dot{\theta}+q\dot{x}), \quad -\frac{\ddot{\phi}}{\dot{\phi}^2} = -\frac{\gamma}{\dot{\phi}}.
\]

{We will assume that \(\dot{\phi}_0 > 0\). Then $\ddot{\phi} = \gamma \dot{\phi}$.}

We now use the symmetry reduction and reconstruction to reobtain the same solutions. Since for the current example the connection coefficients
of the trivial connection vanish, the Lagrange-Poincare-Herglotz equations become
\[
\frac{d}{dt} \left( \frac{\partial l}{\partial w^b} \right) =
\frac{\partial l}{\partial w^a} C^a_{bd} w^d + \frac{\partial l}{\partial s}\frac{\partial l}{\partial w^{a}},
\]
\[
\frac{d}{dt} \left( \frac{\partial l}{\partial \dot{x}} \right) -
\frac{\partial l}{\partial x} = \frac{\partial l}{\partial s}\frac{\partial l}{\partial \dot{x}}.
\]
and substituting for the Lagrangian function we have that
\[
\dot{w}_1 + q\ddot{x} = -1 -\gamma (w_{1} + q\dot{x}), \quad \dot{w}_2 = -w_1w_2 + \gamma w_{2}, \quad q\dot{w}_1 + \ddot{x} = -\gamma (\dot{x} + q w_{1}).
\]
%If we set \( w_1(0) = \dot{\theta}_0 \) and \( w_2(0) = \exp(-\theta_0)\dot{\phi}_0 \), the solution of the above equations is
%\[
%x(t) = -\frac{1}{2} \frac{qt^2}{q^2 - 1} + \dot{x}_0 t + x_0, \quad w_1(t) = \frac{t}{q^2 - 1} + \dot{\theta}_0,
%\]
%\[
%%\]
%Clearly, \( x(t) \) has the desired form.

From the reduced solution \(\check{v}(t) = (x(t), \dot{x}(t), w_1(t), w_2(t), s(t))\),
we could determine the remaining coordinates \((\theta(t), \phi(t))\) directly from the relations \(\dot{\theta} = w_1\) and
\(\exp(-\theta)\dot{\phi}= w_2\). However, let us reproduce the reconstruction process in detail: first we calculate the horizontal lift of \(\check{v}(t)\), and then we use it in the reconstruction
equation. In this way, we will see the effect of the connection in the reconstruction
equation.

The coordinates of the horizontal lift using the generalized mechanical connection are \(\check{v}^H(t) =
(x(t), \phi_H(t), \theta_H(t), \dot{x}(t), w_1(t), w_2(t), s(t))\), with respect to the invariant basis. They can be determined
by the relation \(
0 = \Omega(\dot{\check{v}}^H(t)) \), which in coordinates reads
\[
\dot{\theta}_H = -q \dot{x}, \quad \dot{\phi}_H = -q \dot{x} \phi_H + q \dot{x} \phi_H = 0.
\]
Therefore, \(\theta_H(t) = -qx(t) + qx_0 + \theta_0\) and \(\phi_H(t) = \phi_0\).
Now we determine the curve \(g(t) = (\theta_1(t), \phi_1(t))\) in \(G\) such that \(v = g \check{v}^H\) is the solution of
the Euler-Lagrange-Herglotz equations with the given initial values. This curve is the solution through
the identity of the reconstruction equation \(g^{-1}\dot{g} = 
\Omega(\Gamma(\check{v}^H))\). The left-hand side is \(\dot{\theta}_1 E_1 +
\exp(-\theta_1)\dot{\phi}_1 E_2\). The right-hand side is \((w_1 + q \dot{x})\tilde{E}^C_1 \circ \check{v}^H +(w_2 + q \phi_H \dot{x})\tilde{E}^C_1 \circ \check{v}^H\). The
reconstruction equations are therefore
\[
\dot{\theta}_1 = w_1 + q \dot{x}, \quad \exp(-\theta_1)\dot{\phi}_1 = -\phi_H w_1 + \exp(\theta_H)w_2 - q \phi_H \dot{x}.
\]

%Solving the above equations for \((\theta_1, \phi_1)\) gives
%%\[
%\theta_1(t) = -\frac{1}{2} t^2 + (q \dot{x}_0 + \dot{\theta}_0)t, \quad \phi_1(t) = \dot{\phi}_0 t + \phi_0 \left(%
%1 - \exp\left( \frac{1}{2} (2q \dot{x}_0 - t + 2\dot{\theta}_0)t \right)
%\right).
%\]
%The final solution is therefore indeed
%\[
%\theta(t) = \theta_1(t) + \theta_H(t) = \frac{1}{2} \frac{t^2}{q^2 - 1} + \dot{\theta}_0 t + \theta_0, \quad \phi(t) = \exp(\theta_1(t))\phi_H(t) + \phi_1(t) = \dot{\phi}_0 t + \phi_0.
%\]

Thus, we may write the solution \(v(t) =
(x(t), \phi(t), \theta(t), \dot{x}(t), w_1(t), w_2(t), s(t))\) with respect to the invariant basis, where $\phi(t)$ and $\theta(t)$ are of the form $\theta(t)=\theta_{1}(t)+\theta_{H}(t)$ and $\phi(t)= \phi_{1}(t) + \exp{(\theta_{1}(t))}\phi_{H}(t)$, from where we deduce that $\dot{\theta}=w_{1}$ and $\dot{\phi}=\exp(\theta)w_{2}$.

\section{Conclusions and further work}

In this paper we have developed a process for the reduction of Lagrangian contact systems in the presence of a group of symmetries. We have used the quasi-coordinate method and adapted local references to obtain the corresponding reduced equations, the so-called Lagrange-Poincar\'e-Herglotz equations. One of the necessary geometrical constructions is the prolongation of a connection in a principal fiber bundle using the Hessian of the Lagrangian.

In a future work, we propose the following objectives directly related to the research investigated in this paper:

\begin{itemize}

\item Develop a reduction procedure for discrete contact Lagrangian systems with symmetries.

\item Explore the analog results of the present paper to the case of contact Lagrangian systems subject to nonholonomic constraints and symmetries.

\item Extend the above results for higher-order contact Lagrangian systems with symmetries.

\end{itemize}

\section*{Acknowledgements}

 Modesto Salgado and Silvia Souto acknowledge financial support from Grant  Ministerio de Ciencia, Innovación y Universidades (Spain), project PID2021-125515NB-C21. A. Anahory Simoes, L. Colombo and M. de Le\'on acknowledge financial support from the Spanish Ministry of Science and Innovation, under grants PID2019-106715GB-C21, PID2022-137909NB-C2 and the Severo Ochoa Programme for Centres of Excellence in R\&D (CEX2019-000904-S).

\end{document}